\documentclass{article}
\usepackage{fullpage}
\usepackage{amssymb,amsmath,amsfonts,amsthm}
\usepackage{graphicx}
\usepackage{ifpdf}
\usepackage{subfigure}
\usepackage{times}
\ifpdf\fi
\usepackage{url}
\usepackage{textcomp}
\usepackage[noend,ruled]{algorithm2e}
\usepackage{algorithmicx}

\newtheorem{theorem}{Theorem}
\newtheorem{lemma}[theorem]{Lemma}
\newtheorem{definition}{Definition}

\newtheorem{example}{Example}

\newcommand{\para}[1]{ \medskip \noindent {\bf #1}}
\newcommand{\eat}[1]{}

\newcommand{\A}{{\cal A}}
\newcommand{\D}{{\cal D}}
\newcommand{\X}{{\cal X}}

\newcommand{\eps}{{\varepsilon}}

\title{Mining Frequent Graph Patterns with Differential Privacy}
\author{Entong Shen\\
North Carolina State University\\
eshen@ncsu.edu
 \and Ting Yu\\
 North Carolina State University\\
tyu@ncsu.edu}
\date{}

\begin{document}
\maketitle

\begin{abstract}
	Discovering frequent graph patterns in a graph database offers valuable information
	in a variety of applications. However, if the graph dataset contains
	sensitive data of individuals such as mobile phone-call graphs and web-click graphs,
	releasing discovered frequent patterns may present a threat to the privacy of
	individuals.
	{\em Differential privacy} has recently emerged as the {\em de
	facto} standard for private data analysis due to its provable
	privacy guarantee.
	In this paper we propose the first differentially private algorithm
	for mining frequent graph patterns.
	\vspace{0.5em}
	
	We first show that previous techniques on differentially private discovery of frequent
	{\em itemsets} cannot apply in mining frequent graph patterns due to the 
	inherent complexity of handling structural information in graphs.
	We then address this challenge by proposing a Markov Chain Monte Carlo (MCMC)
	sampling based algorithm.
	Unlike previous work on frequent itemset mining,
	our techniques do not rely on the
	output of a non-private mining algorithm. 
	Instead, we observe that both frequent graph pattern mining and the guarantee of
	differential privacy can be unified
	into an MCMC sampling framework. In addition, we establish the privacy and
	utility guarantee of our algorithm
	and propose an efficient neighboring pattern counting technique as well.
	Experimental results show that the proposed algorithm is able to
	output frequent patterns with good precision.

\end{abstract}

\section{Introduction}\label{sec:intro}
Frequent graph pattern mining (FPM) is an important topic in data
mining research. It has been increasingly applied in a variety of application
domains such as bioinformatics, cheminformatics and social network
analysis.
Given a graph dataset $\D=\{D_1,D_2,\dots,D_n\}$, where each $D_i$ is a graph, 
let $gid(G)$ be the set of IDs of graphs
in $\D$ which contain $G$ as a subgraph.
$G$ is a frequent pattern if its count $|gid(G)|$ (also called {\em support})
is no less than a user-specified
{\em support threshold} $f$.
Frequent subgraphs can
help the discovery of common substructures, and
are the building blocks of further analysis, including graph classification,
clustering and indexing. 
For instance, discovering frequent patterns in social interaction graphs 
can be vital to understand functioning of
the society or dissemination of diseases.

Meanwhile, publishing frequent graph patterns may impose potential 
threat to privacy, if the graph dataset
contains sensitive information of individuals. In many applications, identities are
associated with individual graphs (rather than nodes or edges) which
are considered private. For example,
the click stream during
a browser session of a user is typically a sparse subgraph of the underlying
web graph;
in location-based services, a database may consist of a set of trajectories,
each of which corresponds to the locations of an individual in a given
period of time.
Other scenarios of frequent pattern mining with sensitive graphs may include
mobile phone call graphs \cite{nanavati2006structural} and
XML representation of profiles of individuals.
Therefore, extra care is needed when mining
and releasing frequent patterns in these graphs to prevent leakage of private information of individuals.

\eat{
or similar syntactic models \cite{Machanavajjhala06,ntv-tcpbk-07}
may leak sensitive information
when the anonymized data are combined with an adversary's background knowledge \cite{ganta2008composition}.
Since most subgraph mining algorithms are threshold-based and deterministic,
in the presence of
a strong attacker who can manipulate the input graphs, the attacker
is able to
infer arbitrary information about the target individual.
Below we show such an example:
\begin{example}\label{example1}
Consider a web click graph dataset and an attacker Alice
who is able to actively influence the input graphs to
FPM algorithm $\mathcal{A}$. Alice wants to find out whether the (unknown)
click graph of a user Bob contains a specific pattern $G$.
Suppose $\mathcal{A}$ has
a support threshold of $f$, then Alice can feed $\mathcal{A}$ with
$f-1$ fake graphs having pattern $G$ along with Bob's graph. By observing
whether the output of $\mathcal{A}$ contains $G$, Alice can infer whether
$G$ is contained in Bob's click graph.
\end{example}
}

It has been well recognized that simple anonymization schemes that only remove
obvious identifiers carry serious risks to privacy.
Even privacy-preserving graph mining techniques (e.g.\cite{Li2011})
based on $k$-anonymity \cite{sweeney02-2} are now often considered to offer
insufficient privacy under strong attack models.
Recently, the model of differential privacy \cite{Dwork2006a}
was proposed to restrict the inference of
private information even in the presence of a strong adversary.
It requires that
the output of a differentially private algorithm is
nearly identical (in a probabilistic sense), whether or not a
participant contributes her data to the dataset.
For the problem of frequent graph mining, it means that even an adversary
who is able to actively influence the input graphs cannot infer whether a specific
pattern exists in a target graph.
Although tremendous progress has been made in processing {\em flat} data
(e.g. relational and transactional data)
in a differentially private manner,
there has been very few work on differentially private analysis of graph data,
due to the inherent complexity in handling the structural information in graphs.

In this paper we propose the first algorithm for privacy-preserving mining
of frequent graph patterns that guarantees differential privacy.
Recently several techniques \cite{Thakurta2010,li2012privbasis}
have been proposed to publish frequent {\em itemsets} 
in a transactional database in a differentially private manner. 
It would seem attractive to adapt those techniques to 
address the problem of frequent {\em subgraph}\footnote{We use `graph pattern' and
`subgraph' interchangeably in this paper.} mining.
Unfortunately, compared with private frequent itemset mining, the 
private FPM problem imposes much more challenges. 
First, graph datasets do not have a set of well-defined dimensions (i.e.,{\em items}),
which is required by the techniques in \cite{li2012privbasis}.
Second, counting graph patterns is much more difficult than counting itemsets
(due to graph isomorphism), which makes the size of the output space not
immediately available in our problem. This prevents us from applying the techniques
in \cite{Thakurta2010}.
We will explain the distinction
between \cite{Thakurta2010,li2012privbasis} and our work with more details
in Section~\ref{subsec:challenges}.

\para{Contributions.}
The major contributions of this paper are summarized as follows:
\begin{enumerate}
\item For the first time, we introduce a differentially private algorithm
for mining frequent patterns in a graph database. 
Our algorithm, called {\em Diff-FPM},
makes novel use of a Markov Chain Monte Carlo (MCMC) random walk method to
bypass the roadblock of an output space with unknown size. This enables us to
apply the {\em exponential mechanism}, which is a general approach to achieving
differential privacy.
Moreover, unlike \cite{Thakurta2010} that relies on the output of
a non-private itemset mining algorithm, our technique integrates the process
of graph mining and privacy protection as a whole. This is due to the
observation that both frequent pattern mining and the application of exponential mechanism
can be unified into an MCMC sampling framework.
\item Our approach provides provable privacy and utility guarantee on the output
of our algorithm. We first show that our algorithm gives $(\eps,\delta)$-differential
privacy, which is a relaxed version of $\eps$-differential privacy. 
We then show that when the random walk has reached its steady state, {\em Diff-FPM}
gives $\eps$-differential privacy. For utility analysis,
because a private frequent graph mining algorithm usually does not output the exact answer,
we quantify the quality of our result by providing a high-probability upper
bound on how far the support of the reported patterns can be from 
the support threshold specified by the user.
\item The most costly operation in our algorithm is counting the support
of a pattern in the graph dataset, due to the fact that subgraph isomorphism
test is NP-complete. In order to propose more efficiently a neighboring pattern in MCMC
sampling, we develop optimization techniques that significantly reduce the
number of invocations to the subgraph isomorphism test subroutine. 
\item We conduct an extensive experimental study on the effectiveness and efficiency
of our algorithm. With moderate amount of privacy budget, {\em Diff-FPM}
is able to output
private frequent graph patterns with at least 80\% precision.
\end{enumerate}

The paper is organized as follows: 
The basic concept and techniques for differential privacy, as well as a formal
definition of the FPM problem are introduced
in Section~\ref{sec:prelim}.
Section~\ref{sec:algo} and Section~\ref{sec:detail} introduces our {\em Diff-FPM} algorithm, 
whose privacy and utility
analysis is provided in Section~\ref{sec:analysis}. 
The experiment result is presented in
Section~\ref{sec:exp}. We review related work in Section~\ref{sec:related} and
Section~\ref{sec:conclusion} concludes our discussion.
\section{Preliminaries}\label{sec:prelim}

\subsection{Frequent Graph Pattern Mining}\label{subsec:formulate}
Frequent graph pattern mining (FPM) aims at discovering the subgraphs
that frequently appear in a graph dataset. Formally,
Let $\D=\{D_1,D_2,\dots,D_n\}$ be a sensitive graph database
which contains a multiset of graphs. Each graph $D_i\in \D$ has
a unique identifier. Let $G=(V,E)$ be a (sub)graph pattern, the graph
identifier set $gid(G)=\{i:G \subseteq D_i \in \D\}$ includes
all IDs of graphs in $\D$ that contain a subgraph isomorphic
to $G$. We call $|gid(G)|$ the \textit{support} of $G$ in $\D$.
The FPM algorithm can be defined either as returning all subgraph patterns whose supports
are no less than a user-specified threshold $f$, or as returning the top $k$ frequent
patterns given an integer $k$ as input. 
One can easily convert one version to the other.

All graphs we consider in this paper are undirected, connected and labeled.
Note that each node has a label and multiple nodes can have the same label.
Depending on the application, the patterns considered
may be subject to a set $\mathcal{R}$
of rules which are related to domain knowledge or user specifications.
It is common to place
an upper bound on the number of nodes and/or edges in the patterns, or specify
the set of possible labels.
For example, if the graphs represent chemical compounds, a rule may require
the degree of a vertex labeled `{\em C(arbon)}' be no greater than 4. Another
rule may specify that any output contains at least 5 vertices,
in order to filter out some trivial patterns. 

Many non-private algorithms have been proposed for finding frequent subgraphs.
The most representative approaches include Apriori algorithm \cite{inokuchi2000apriori}
and the gSpan \cite{yan2002gspan} algorithm. The Apriori algorithm exploits the 
observation that if a graph pattern $G$ is frequent, all its subgraphs must also
be frequent. The algorithm works by exploring the search space, i.e., generating
candidate patterns and pruning infrequent ones. The gSpan algorithm 
maps each graph to a unique minimum DFS code,
which skips the candidate generation process. For a detailed review of graph pattern
mining and other related work, please refer to Section~\ref{sec:related}.

\subsection{Differential Privacy}\label{subsec:dp}
Differential privacy \cite{Dwork2006a} is a recent privacy model which provides strong
privacy guarantee. Informally, a data mining or publishing procedure
is differentially private if the outcome is insensitive to any particular
record in the dataset. In the context of graph pattern mining, let
$\D,\D'$ be two {\em neighboring datasets}, i.e., $\D$ and $\D'$ differ
in only one graph, written as $|| \D - \D' || = 1$. Let
$\D^n$ be the space of graph datasets containing $n$ graphs.
\begin{definition}[$\eps$-differential privacy]
\label{def:dp1}
A randomized algorithm $\A$ is $\eps$-differentially private if for all
neighboring datasets $\D$,$\D'$ $\in \D^n$, and any set of possible
output $\mathcal{O} \subset Range(\A)$:
$$\Pr[\A(\D)\in \mathcal{O}]\leq e^{\eps}\;\Pr[\A(\D')\in \mathcal{O}].$$
\end{definition}
The parameter $\eps>0$ allows us to control the level of privacy. A smaller
$\eps$ suggests more limit posed on the influence of a single graph.
Typically, the value of $\eps$ should be small ($\eps<1$). $\eps$ is
usually specified by the data owner and referred as the {\em privacy budget}.
In section \ref{subsec:privacy} our discussion is related to a weaker notion called
$(\eps,\delta)$-differential privacy \cite{Dwork2006}, which allows a small additive
error factor of $\delta$.
\begin{definition}[$(\eps,\delta)$-differential privacy]\label{def:dp2}
A randomized algorithm $\A$ is $(\eps,\delta)$-differential private if for all
neighboring datasets $\D$,$\D'$ $\in \D^n$, and any set of possible
output $\mathcal{O} \subset Range(\A)$:
$$\Pr[\A(\D)\in \mathcal{O}]\leq e^{\eps}\;\Pr[\A(\D')\in \mathcal{O}]+\delta.$$
\end{definition}
\para{Laplace Mechanism.} The most common technique for designing differentially
private algorithms is to add random noise to the true output of a function \cite{Dwork2006a}.
The noise is calibrated according to the {\em sensitivity} of the function,
which is defined as the maximum difference in the output for any neighboring
datasets. Formally,
\begin{definition}[Sensitivity]\label{def:sensitivity}
For any function $f : \D^n$$\to \mathbb{R}$, the sensitivity of $f$ is
$$\Delta f = \max_{\D,\D': \|\D-\D'\| = 1 } |f(\D)-f(\D')|.$$
\end{definition}
Given a dataset $\D$ and a numeric function $f$, the Laplace mechanism
achieves $\eps$-differential privacy by releasing $\tilde{f}(\D)=
f(\D)+Lap(\Delta f/\eps)$, where $Lap(\lambda)$ denotes a random
variable drawn from the Laplace distribution with mean of 0 and variance
of $2\lambda^2$.

Applying the Laplace mechanism requires the output of a function being numeric.
In many applications,
however, the output may be models, classifiers or graphs which contain
structural information that are not easily perturbed by the Laplace
mechanism.
Thus it cannot be directly applied to the problem of frequent subgraph mining.
Still, we can use this technique to report
the frequencies of the patterns we output.

\para{Exponential Mechanism.}
A more general technique of applying differential privacy
is the exponential mechanism \cite{McSherry2007}.
It not only
supports non-numeric output but also captures
the full class of differential privacy mechanisms. The exponential mechanism
considers the whole {\em output space} and assumes that each possible
output
is associated with a real-valued utility score. By sampling from a distribution where
the probability of the desired outputs are exponentially amplified, the exponential
mechanism (approximately) finds the desired outputs while ensuring differential privacy.

Formally, given input space $\D^n$ and output space $\X$,
a score function $u:\D^n \times \X \to \mathbb{R}$ assigns
each possible output $x \in \X$ a score $u(\D,x)$ based on the input $\D\in \D^n$.
The mechanism then draws a sample from the distribution on $\X$ which assigns
each $x$ a probability mass proportional to $\exp(\eps u(\D,x)/2\Delta u)$,
where $\Delta u=\max_{\forall x,\D,\D'}|u(\D,x)-u(\D',x)|$ is the sensitivity
of the score function. Intuitively, the output with a higher score is
exponentially more likely to be chosen. It is shown that
this mechanism satisfies $\eps$-differential privacy \cite{McSherry2007}.
\begin{theorem}\label{thm:exp-mech}
\cite{McSherry2007} Given a utility score function $u:\D^n \times \X \to \mathbb{R}$
for a dataset $\D$, the mechanism $\A$,
$$\A(\D,x)\triangleq \textrm{return $x$ with probability} \propto \exp(\frac{\eps u(\D,x)}{2\Delta u})$$
gives $\eps$-differential privacy.
\end{theorem}

The exponential mechanism has been shown to be a powerful technique in finding
private medians \cite{Cormode2012}, mining private frequent itemset 
\cite{Thakurta2010,li2012privbasis} and more
generally adapting
a deterministic algorithm to be differentially private \cite{Mohammed2011}. 
As discussed in Section~\ref{sec:intro}, it is infeasible to find 
frequent graph patterns privately using the Laplace mechanism by 
adding noise to the support of each possible pattern.
Our {\em Diff-FPM} algorithm works by carefully applying the exponential mechanism.
In this process we must overcome several critical challenges, which are identified
next.

\subsection{Challenges}\label{subsec:challenges}

There has been work \cite{Thakurta2010,li2012privbasis} on mining frequent 
{\em itemsets} in a transaction dataset
under differential privacy. 
However, the shift from
transactions to graphs poses significant new challenges, which
make the previous techniques no longer suitable in our problem. 
In \cite{li2012privbasis}, transaction datasets are viewed as high-dimensional
tabular data, and the proposed approach projects the input database onto lower
dimensions. However, graph datasets do not have a well defined set of {\em items},
i.e., dimensions, which renders the approach in \cite{li2012privbasis} inapplicable
in our FPM problem. 
In \cite{Thakurta2010}, two methods are proposed which make use of a notion of
truncated frequency. However, those methods cannot be used in our problem due
to the following fundamental challenges:

\para{Support Counting.} Obtaining the support of a graph pattern is 
much more difficult than counting itemsets. 
An itemset pattern
can be represented by an ordered list or a bitmap of item IDs and does not contain
structural information as in graphs. Checking the existence of an itemset
in a transaction only takes $O(1)$ time (after simple data
structures such as bitmaps have been built), while checking whether a subgraph pattern
exists in a graph is NP-complete due to subgraph isomorphism.

\para{Unknown Output Space.} The output space $\X$ in our problem contains a finite number of graph
patterns which may or may not exist in the input dataset.
Under differential privacy, any pattern in the output
space should have non-zero probability to be in the final output.
The knowledge of the output space is essential in applying the exponential mechanism,
in which we need to sample a pattern $x$ with probability
\begin{equation} \label{eq:exp-mech}
\pi(x) = \frac{\exp(\eps u(x)/2\Delta u)}{C},
\end{equation}
where $C=\sum_{x\in \X}\exp(\eps u(x)/2\Delta u)$ is the normalizing constant
according to Theorem~\ref{thm:exp-mech}.
The most straightforward way to compute $C$ requires enumerating all the patterns in
the output space.
In \cite{Thakurta2010}, a technique is proposed to apply the exponential mechanism
without enumerating if the size of the output space is known.
However, unlike \cite{Thakurta2010}, in which the output space size can
be obtained by simple combinatorics (i.e.,$\binom{m}{l}$ patterns of size
$l$ given an alphabet of size $m$), {\em the size of the output space $\X$ in our
problem is not immediately available (due to graph isomorphism\footnote{
Essentially, we need to answer the question 'Are there any closed form formula or
 polynomial time algorithms to count the number of graphs given the number of
 vertices, edges and a set of possible labels?'. 
 (1) If the vertex labels are all unique, we know the number of graphs given $n$ vertices is $2^{n(n-1)/2}$.
 (2) If the graph is unlabeled, the problem is considerably harder due to graph isomorphism. P\'{o}lya enumeration theorem provides an algorithm to compute the number of isomorphism classes of graphs with $n$ vertices and $m$ edges \cite{polya}. But it gives neither a formula nor a generating function.
 (3) When the labels are not unique, the problem is at least as hard as the unlabeled case.
})}, which prohibits
us from applying exponential mechanism directly. 
Therefore we cannot apply
the same techniques as in \cite{Thakurta2010}.

Given the analysis above, we need to develop
new ways to overcome the issue of an unknown $|\X|$.
Note that although the {\em global} information on the output space is not accessible,
we do have the {\em local} information on any specific pattern -- given any pattern
$x$, we can immediately calculate its utility score $u(x)$ (related to $|gid(x)|$,
see Section~\ref{sec:algo} for details). In addition, the unknown normalizing
constant $C$ is common to all patterns.
That is, given any pair of patterns $x_1,x_2$, the ratio of probability mass
$\pi(x_1)/\pi(x_2)$ is available without knowing the exact probabilities,
according to Eq.(\ref{eq:exp-mech}).
Such scenarios, where one needs to
draw samples from a probability distribution known up to a constant factor,
also arise in statistical physics when analyzing dynamic systems, where
Markov Chain Monte Carlo (MCMC) methods are often used.
Inspired by that,
our idea is to perform a random walk based on locally computed probabilities.
By carefully choosing the neighbor and the probability
of moving in each step using the Metropolis-Hastings (MH) method \cite{rubinstein2008simulation}, 
the random walk will converge to the target distribution, from which
we can output samples.
Next we discuss the details of our {\em Diff-FPM} algorithm.

\section{Private FPM Algorithm}\label{sec:algo}

\subsection{Overview}\label{sec:overview}

The key challenge of handling graph datasets is the
unknown output space when applying the exponential mechanism. 
The {\em Diff-FPM} algorithm meets the challenge by 
unifying frequent pattern mining and applying differential
privacy into an MCMC sampling framework. 
The main idea of {\em Diff-FPM} is to simulate a Markov chain by 
performing an MCMC random walk in the output space.
Our goal is that when the random walk reaches its steady state,
the stationary distribution of the Markov chain matches the
target distribution $\pi$ in Eq.(\ref{eq:exp-mech}).
In Section~\ref{sec:MH} we will explain in detail how to apply the
Metropolis-Hastings (MH) method in our problem to achieve
this goal. 
Before that, we need to define the state space in which we perform the random walk. 


\para{Partial Order Full Graph.} To facilitate the MH-based random walk in
the output space, we define the Partial Order Full Graph (POFG) as
the state space of the Markov chain on which the sampling algorithm
run the simulation. Each node in
POFG corresponds to a unique graph pattern and each edge in POFG represents
a possible `extension' (add or remove one edge) to a neighboring pattern.
Naturally, each node in the POFG has three types of neighbors: {\em sub-neighbor}
(by removing an edge), {\em super-backward neighbor} (by connecting two existing
nodes) and {\em super-forward neighbor} (by adding and connecting to a new node).

\begin{example}
Figure~\ref{fig:example} shows a simple graph dataset containing 3 graphs
and its POFG. The dashed patterns have support smaller than 2 in the dataset.
Pattern $A-A-C$ has two sub-neighbors, one super-backward neighbor and several
super-forward neighbors (only one shown in Figure~\ref{fig:POFG}). Self-loops
and multi-edges are not considered in this example and thus are excluded from
the output space.
\end{example}

At a higher level, 
the random walk starts with an arbitrary pattern and proceeds to an adjacent 
pattern with certain probability in each step.
Since the transition decision is made
solely based on local information, there is no need to construct the global POFG
explicitly. 
When the random walk has reached its steady state, the probability of being
in state $x$ follows exactly the target distribution $\pi(x)$ in Eq.(\ref{eq:exp-mech}). 
Then the current state is drawn as a sampled pattern. Since the frequent patterns
have larger probabilities in the target distribution, they are more likely to appear
in the final output.

Before introducing the details,
we need to make sure that the random walk on POFG we design
indeed converges to a stationary distribution.
A random walk needs to be finite, irreducible, and aperiodic to converge to
a stationary distribution \cite{rubinstein2008simulation}.
The analysis is similar to that in \cite{AlHasan2009}.

\subsection{Detailed Descriptions}

\subsubsection{Backgrounds on Markov Chain}
\label{sec:appx-markov}
A Markov chain is a discrete-time stochastic process defined over a set of
states $\X$. $\X$ can be finite or countably infinite.
The Markov property
requires that given the present state, the past and the future are independent.
The stochastic process is characterized by the \textit{transition matrix} $P$,
which defines the probability of transition between any two states in $\X$,
i.e., $P(x,y)$ is the probability that the next state will be $y$, given that
the current state is $x$. For all $x,y \in \X$, we have $0\leq P(x,y)\leq 1$,
and $\sum_{y}P(x,y)=1$, i.e., $P$ is row-stochastic.

A stationary distribution of a Markov chain with transition probability $P$
is a probability distribution $\pi$ (a row vector of size $|\X|$), such that
$\pi = \pi P$.

If a Markov chain is finite, irreducible and aperiodic,
regardless of where it begins, the chain will converge to the
stationary distribution. We also say it has reached the {\em steady state}
when the chain has converged.

If the state space $\X$ of a Markov chain is the set $\mathcal{V}$ of 
a graph $\mathcal{G}=(\mathcal{E},\mathcal{V})$,
and if for any $u,v\in \mathcal{V}$, $(u,v)\notin \mathcal{E}$ implies $P(u,v)=0$, then the process
is also called a {\em random walk} on the graph $\mathcal{G}$. In other words, transitions
only occur between adjacent nodes. 

\subsubsection{Applying the MH method}
\label{sec:MH}
\begin{figure}[t]
\centering
\subfigure[Graph database with 3 graphs]{\label{fig:database}
\includegraphics[width=0.3\columnwidth]{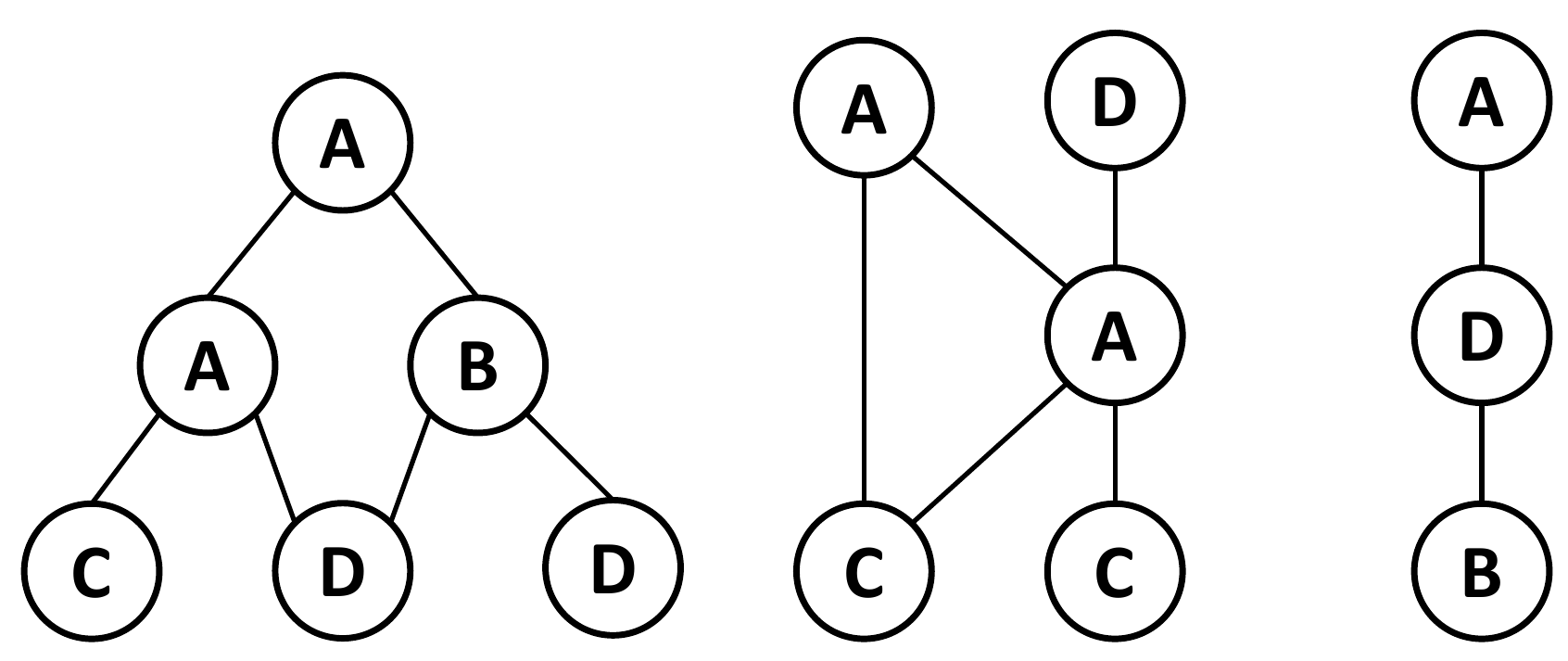}
}\\
\subfigure[Part of POFG of Figure~\ref{fig:database}]{\label{fig:POFG}
\includegraphics[width=0.5\columnwidth]{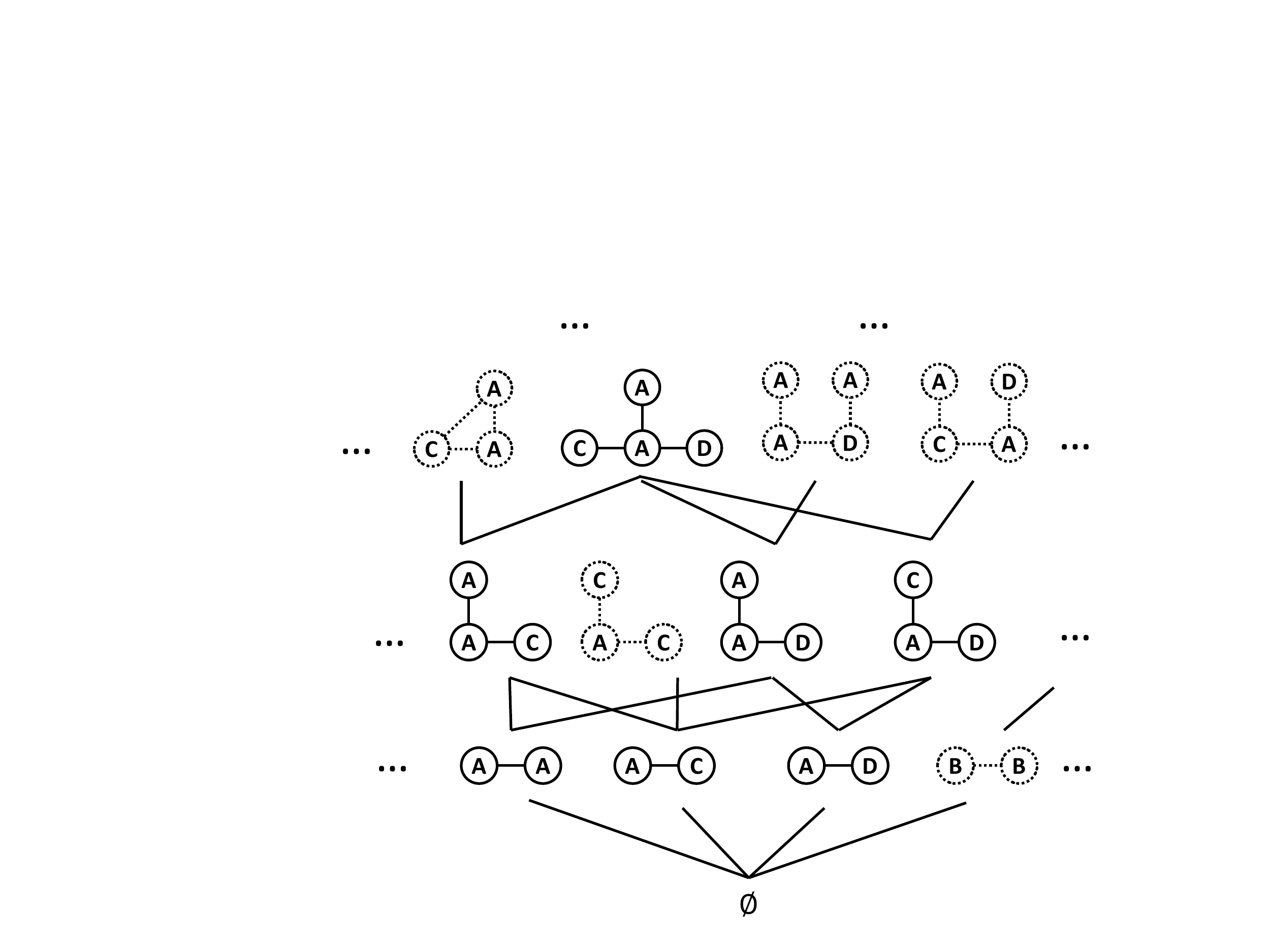}
}
\caption{Example graph database and POFG}
\label{fig:example}
\end{figure}

The MH method is a
Markov Chain Monte Carlo (MCMC) method for obtaining a sequence of random samples
from a target probability distribution for which direct sampling is difficult.
It only requires that a function proportional to the probability mass be calculable.
The main idea of the MH method is to simulate a Markov chain such
that the stationary distribution of the chain matches the target distribution
\cite{rubinstein2008simulation}.

Suppose we want to generate a random
variable $X$ taking values in $\X=\{x_1,\dots,x_{|\X|}\}$, according
to a target distribution $\pi$, with
\begin{equation}
\pi(x_i)=\frac{b(x_i)}{C},\qquad x_i \in \X \nonumber
\end{equation}
where all $b(x_i)$ are strictly positive, $|\X|$ is large, and the normalizing
constant $C=\sum_{i=1}^{|\X|}b(x_i)$ is difficult to calculate. The MH method
first constructs an $|\X|$-state Markov chain $\{X_t,t=0,1,\dots\}$ on
$\X$ whose evolution relies on an arbitrary proposal transition
matrix $Q=\big(q(x,y)\big)$ in the following way:
\begin{itemize}
\item When $X_t=x$, generate a random variable $Y$ satisfying $P(Y=y)=q(x,y)$,
$y\in \X$
\item If $Y=y$, let
\begin{equation}
X_{t+1} = \left\{
\begin{array}{rl}
y & \text{ with probability } \alpha_{xy},\\
x & \text{ with probability } 1-\alpha_{xy},
\end{array} \right.\nonumber
\end{equation}
\end{itemize}
where $\alpha_{xy}=\min\bigg\{\frac{\pi(y)q(y,x)}{\pi(x)q(x,y)},1\bigg\}=\min\bigg\{\frac{b(y)q(y,x)}{b(x)q(x,y)},1\bigg\}$.
It means that given a current state $x$, the next state is proposed according
to the proposal distribution $Q$. $q(x,y)$ is the probability mass of state
$y$ among all possible states given the current state is $x$. With probability
$\alpha_{xy}$, the proposal is accepted and the chain moves to the new state
$y$. Otherwise it remains at state $x$. It follows that $\{X_t,t=0,1,\dots\}$
has a one-step transition probability matrix $P$:
\begin{equation}
P(x,y)= \left\{
\begin{array}{ll}
q(x,y)\alpha_{xy}, & \text{if } x\neq y\\
1-\sum_{z\neq x}q(x,z)\alpha_{xz}, & \text{if } x=y
\end{array} \right.\nonumber
\end{equation}

It can be shown that for the above $P$, the Markov chain is reversible
and has a stationary distribution $\pi$, equal to the target distribution.
Therefore, once the chain has reached the steady state, the sequence of samples
we get from the MH method should follow the target distribution.
Next we use an example to explain how the state transition works in our 
{\em Diff-FPM} algorithm.

\begin{example}
Consider a random walk on the POFG illustrated in Figure~\ref{fig:POFG}.
Suppose the current state of the walk is `A-A-D' (pattern $x$). 
Following the MH method, one of pattern $x$'s neighbors needs to be proposed according
to a proposal distribution $q(x,y)$. For simplicity, in this example each
neighbor has an equal probability to be proposed, i.e., $q(x,y)=1/|N(x)|$,
where $N(x)$ is the neighbor set of $x$. Assuming `A-D' (pattern y) is proposed and
$|N(x)|=5$, $|N(y)|=10$, $b(\cdot)=\exp(|gid(\cdot)|/2)$,
the probability of accepting the proposal is calculated as
$\alpha_{xy}=$$\min\bigg\{\frac{\exp(3/2)\cdot (1/10)}{\exp(2/2)\cdot (1/5)},1\bigg\}$
$=0.82$. We can then draw a random number between 0 and 1 to decide whether walking
to pattern $y$ or staying at $x$.
\end{example}

The ability to generate a sample without knowing the normalizing constant
of proportionality is a major virtue of the MH method.
This salient feature 
fits perfectly the scenario when direct application of the exponential
mechanism is formidable due to an unmanageable output space.

\SetKwInOut{Input}{input}\SetKwInOut{Output}{output}
\begin{algorithm}[t]
\LinesNumbered
\caption{{\em Diff-FPM} algorithm}
\Input{Graph data set $\D$, support threshold $f$, privacy budget $\eps_1,\eps_2$}
\Output{A set $S$ of $k$ private frequent patterns with noisy supports}
\For{$i=1 ~\textrm{to} ~k$}{
    Choose any pattern in the output space as the initial pattern\;
    \While{True}{
        Propose a neighboring pattern $y$ of current pattern $x$ according to the proposal distribution (Eq. \ref{eqn:prop})\;
        Accept the proposed pattern with probability $\alpha_{xy}=\min\big\{\frac{\exp(\eps_1 u(y)/2k\Delta u)q_{yx}}{\exp(\eps_1 u(x)/2k\Delta u)q_{xy}},1\big\}$\;
        \If{convergence conditions are met}{
             Add current pattern to $S$ and remove it from the output space\;
            \textbf{break}\;
        }
    }
}
(Optional) for each pattern in $S$, perturb its true support by Laplace mechanism with privacy budget $\eps_2/k$\;
\end{algorithm}


The description of the {\em Diff-FPM} algorithm above can be summarized in Algorithm 1.
The input consists of the raw graph dataset $\D$, a support threshold $f$ and
the privacy budget $\eps = \eps_1 + \eps_2$. If the top-$k$ frequent patterns are desired,
we first run non-private FPM algorithms such as gSpan \cite{yan2002gspan} to get the
support threshold $f$, i.e., the support of the $k$th frequent
pattern. If one only needs $k$ patterns whose supports are no less than
a threshold, $f$ can be directly provided to the algorithm. At a higher level,
Algorithm 1 consists of two phases: sampling and perturbation. The sampling phase
includes $k$ applications of the exponential mechanism via
MH-based random walk in the output space.

Initially, we select an arbitrary pattern
in the output space to start the walk (Line 2). At each step, we propose a
neighboring pattern $y$ of the current pattern $x$
according to a {\em proposal distribution} (Line 4).
The proposal distribution does not affect the correctness of the MH method,
so we defer the details to Section~\ref{sec:proposal} (it does affect the speed of convergence
though). The proposed pattern is then accepted with probability $\alpha_{xy}$ as
in the MH-algorithm (Line 5), where $u(\cdot)$ is the score function with $\Delta u$
being the sensitivity of $u(\cdot)$. We explore the design space of the score function
in the next paragraph. When the Markov chain has converged
(see Section~\ref{sec:converge} for convergence diagnostic),
we output the current pattern and remove it from the output space (Line 6 to 8).
We then start a new walk until $k$ patterns have been sampled. Finally, if
one wants to include the support of each output pattern as well,
the count of each pattern is perturbed by adding $Lap(k/\eps_2)$ noise (Line 9).

\subsubsection{Score Function Design} 
Choosing the utility score function is
vital in our approach as it directly affects the target distribution.
A general guideline is that the patterns with higher supports should have higher
utility scores in order to have larger probabilities to
be chosen according to exponential mechanism.
Under this guideline, given an input database $\D$, the most straightforward choice is to
let $u(x,\D)=|gid(x)|$ for any pattern $x$. In this case, the sensitivity $\Delta u$
is exactly 1 since the support of any subgraph pattern may vary by at most 1
with the addition or removal of a graph in the dataset. Other choices include
assigning the same utility scores to all patterns having supports no less than $f$,
or deliberately lowering the scores of the infrequent patterns. For example,
let $u(x)=a(|gid(x)|-b)$ if $|gid(x)|<f$, where $0<a<1$, $b>0$, and $u(x)=|gid(x)|$
if $|gid(x)|\geq f$. In this case, the infrequent patterns have even less
probability to be sampled. However, this will also increase $\Delta u$
and thus deteriorate the utility, according to Theorem~\ref{thm:exp-mech}.
We will further study the impact of various score functions in the experiment section.

\subsubsection{Proposal Distribution}\label{sec:proposal}
Although in theory the proposal distribution can be
arbitrary, it can essentially impact the efficiency of the MH method
by affecting the mixing time (time to reach steady state).
A good proposal distribution can improve the convergence speed by
increasing the accept rate $\alpha_{xy}$ in the MH method.
On the contrary, if the proposed pattern is often rejected, the chain can
hardly move forward.
It has been suggested that
one should choose a proposal distribution close to the target distribution \cite{gilks1996markov}.
In our problem setting, it is preferable to make a distinction between
the patterns having support
no less than $f$ (referred as {\em frequent} patterns)
and those whose supports are lower (referred as {\em infrequent} patterns).
Given a current state $x$, we denote the set of {\em frequent} neighbors
of $x$ as $N_1(x)$ and the set of {\em infrequent} neighbors as $N_2(x)$.
Since $|N_2(x)|$ is usually larger than $|N_1(x)|$, we will balance the
probability mass assigned to $N_1(x)$ and $N_2(x)$ by introducing a tunable
parameter $\eta$. For the same reason, we use $\rho$ to control the
bias toward either the sub-neighbors $N_1^{b}(x)$ or the super-neighbors
$N_1^{p}(x)$ within the desired set $N_1(x)$. Our heuristic based
proposal distribution is formally described below:
\begin{equation}
\label{eqn:prop}
Q(x,y)= \left\{
\begin{array}{ll}
\rho\eta\times \frac{1}{|N_1^{b}(x)|}, & \text{if } y\in N_1^{b}(x)\\
(1-\rho)\eta\times \frac{1}{|N_1^{p}(x)|}, & \text{if } y\in N_1^{p}(x)\\
(1-\eta)\times \frac{1}{|N_2(x)|}, & \text{if } y\in N_2(x)
\end{array} \right.
\end{equation}
The best values of $\eta$ and $\rho$ can be tuned experimentally.
If any of the three sets of neighbors in Eq.\ref{eqn:prop} is empty, its
probability mass will be re-distributed (by setting $\rho=0$, $\rho=1$ and $\eta=1$ respectively).

\subsubsection{Pattern Removal} 
In line 6 to 8 of Algorithm 1, after the convergence
conditions are met and a sample pattern $g$ is outputted, we need to exclude
$g$ from the output space by connecting $g$'s neighbors and removing
$g$ in the POFG. In our implementation this is done by replacing $g$ by all the
neighbors of $g$ whenever $g$ appears in some pattern's neighborhood.
Note that we do not output multiple patterns when the chain has converged.
This is because once a pattern is sampled, it should be excluded from the
output space and thus have zero probability to be chosen. Therefore
adjustment to the output space is necessary after each sample. For the same
reason we do not run multiple chains at once.

\subsection{Convergence Diagnostics}\label{sec:converge}
The theory of MCMC sampling requires that samples are drawn when the Markov
chain has converged to the stationary distribution, which is also our target
distribution $\pi$. The most straightforward way to diagnose convergence is to
monitor the distance between the target distribution $\pi$ and
the distribution of samples $\hat{\pi}$.
In practice, however, $\pi$ is often known only up to a constant factor. To deal with this problem, several online
diagnostic tests have been developed in the MCMC literature \cite{gilks1996markov}
and used in random walk based sampling on graphs \cite{gjoka2010walking}.

Online diagnostics rely on detecting whether the chain has lost its
dependence on the starting point. In particular, two standard convergence
tests Geweke diagnostic \cite{geweke1992evaluating} and Gelman and
Rubin diagnostic \cite{gelman1992inference} are commonly used,
which are based on analysis of intra-chain and inter-chain properties respectively.
Since our problem setting does not support running multiple chains
at the same time, we will focus on the Geweke diagnostic.

The Geweke diagnostic takes two non-overlapping parts
(usually the first 0.1 and last 0.5 proportions) of the Markov chain and
compares the means of both parts to see
if they are from the same distribution.
Specifically, let $X$ be a sequence of samples of our metric of interest
and $X_1,X_2$ be the two non-overlapping subsequences. Geweke computes the
$Z$-score: $Z=\frac{E(X_1)-E(X_2)}{\sqrt{Var(X_1)+Var(X_2)}}$. With increasing
number of iterations, $X_1$ and $X_2$ should move further apart and become less
and less correlated. When the chain has converged, $X_1$ and
$X_2$ should be identically distributed
with $Z\sim N(0,1)$ by law of large numbers. We can declare convergence
when $Z$ has continuously fallen in the $[-1,1]$ range. Since the samples
in our problem are graph patterns rather than a scalar, we may need to monitor
multiple scalar metrics related to different properties of the sampled pattern
and declare convergence when all these metrics have converged.

We need to acknowledge that these convergence diagnostic tools from the MCMC literature are
heuristic {\em per se}. Verifying the convergence
remains an open problem if the distribution of samples is not directly observable.
Even so, {\em Diff-FPM} still achieves $(\eps,\delta)$-differential privacy if there exists
a small distance between the target and simulation distributions, as we will show in
Lemma~\ref{thm:error} in Section~\ref{sec:analysis}.

\section{Efficient Exploration of Neighbors (EEN)}\label{sec:detail} \label{sec:een}
We have discussed so far the core of the {\em Diff-FPM} algorithm
and seemingly it could be run straightforwardly. However,
without certain optimization, the computation cost 
might render the algorithm
impractical to run.
The most costly operation in the {\em Diff-FPM} algorithm is proposing a neighbor
of the current pattern $x$. According to the proposal distribution in Eq.\ref{eqn:prop},
this requires knowledge on the support of each pattern in $x$'s neighborhood.
Due to the fact that subgraph isomorphism test is NP-complete,
obtaining the support of each neighbor might become a computation bottleneck. 
To overcome this problem, we have developed an efficient algorithm (called EEN),
which aims at
minimizing the
number of invocations to the subgraph isomorphism test subroutine.
Experimental result in Section~\ref{sec:exp} shows that the time cost per iteration
can be reduced by up to an order of magnitude using this optimization. 

\subsection{Problem Formulation}
In order to propose a neighbor $y$ of a pattern $x$ according to the proposal
distribution, we need
to investigate the neighbor set $N(x)$ of $x$ and test the frequentness of each neighbor 
$y\in N(x)$.
The task of neighbors exploration can be described as: given a pattern $x$,
find the set of frequent
sub-neighbors $N_1^b(x)$, frequent super-neighbors $N_1^p(x)$ and
infrequent neighbors $N_2(x)$, as introduced in the proposal distribution (see Eq.\ref{eqn:prop}).

The neighbor set $N(x)$ is composed of two parts - super-neighbors
$N^{p}(x)$ and sub-neighbors $N^{b}(x)$.
A pattern $y$ is a super-neighbor of $x$ if $y=x\diamond e$
and $x\subset y$ (we use $\subset$ to denote subgraph relationship), 
where $e$ is a new edge and $\diamond$ is an extension operation.
If $e$ connects two existing nodes in $x$, it is called a back edge. Otherwise,
a new node is created with a random label from a label set $\mathcal{L}$ and
then connected to an existing node in $x$. In this case the new edge is called a forward edge.
Thus $N^{p}(x)=N_{back}^p(x)\cup N_{fwd}^p(x)$, where $N_{back}^p(x)$ and $N_{fwd}^p(x)$
are the sets of super-back and super-forward neighbors of $x$ respectively.

Similarly, pattern $y$ is a sub-neighbor of $x$ if $x=y\diamond e$ and $y\subset x$.
There are two types of edge removals as well.
Back edge removal removes an edge and keeps the remaining pattern connected with no vertex
removed, while forward edge removal isolates exactly one vertex which is also removed
from the resulting pattern. The above neighbors generation process ensures the
random walk is reversible (which is sufficient for the chain to have a stationary distribution), i.e., for any neighboring patterns $x$ and $y$,
if there is a walk from $x$ to $y$, $y$ can also walk to $x$ and vice versa.

\begin{algorithm}[ht!]
\LinesNumbered
\SetAlgoVlined
\SetKwFunction{FnSubFreq}{\sc sub\_is\_freq}
\SetKwData{True}{true}
\SetKwData{False}{false}
\caption{The EEN algorithm}
\Input{Pattern $x$, graph dataset $\D$, support threshold $f$}
\Output{$N_1^b(x)$,$N_1^p(x)$,$N_2(x)$}
Initialize $N_1^b$,$N_1^p$,$N_2\gets \emptyset$ ($x$ omitted for brevity)\;
Find membership bitmap $B_x$ using VF2 isomorphism test\;
Populate sub-neighbors $N^b$, super-back neighbors $N_{back}^p$, super-forward neighbors $N_{fwd}^p$\;
\tcc{\small Explore sub-neighbors $N^b$}
\lIf{$sum(B_x)\geq f$}{$N_1^b\gets N_1^b\cup N^b$\;}
\lElse{
    \For{$x'\in N^b$}{
        \lIf{{\sc sub\_is\_freq} $(x',B_x)$}{$N_1^b\gets N_1^b \cup \{x'\}$\;}
        \lElse{$N_2\gets N_2 \cup \{x'\}$\;}
    }
}
\tcc{\small Explore super-back neighbors $N_{back}^p$}
\lIf{$sum(B_x)< f$}{$N_2\gets N_2\cup N^p_{back}$\;}
\Else{
    $\forall x'\in N^p_{back}$, initialize dictionary $H[x']=0$\;
    \For{$i\gets 1$ to $|\D|$}{
        Find set $\mathcal{M}$ of all mappings between $D_i$ and $x$\;
        \For{$x'\in N^p_{back}$}{
            \If{$H[x']<f$ and $|\D|-i+H[x']\geq f$}{
                Let $(u,v)$ be the back edge, i.e., $x=x'\diamond(u,v)$\;
                \For{$m\in \mathcal{M}$}{
                    \If{$m(u),m(v)$ are adjacent in $D_i$}{
                    $H[x']\gets H[x']+1$\;
                    break\;
                    }}}}}
    \For{$x'\in N^p_{back}$}{
        \lIf{$H[x']\geq f$}{$N_1^b\gets N_1^b\cup \{x'\}$\;}
        \lElse{$N_2\gets N_2\cup \{x'\}$\;}
        }
    }
\tcc{\small Explore super-forward neighbors $N_{fwd}^p$}
\lIf{$sum(B_x)< f$}{$N_2\gets N_2\cup N^p_{fwd}$\;}
\Else{
    $\forall x'\in N^p_{fwd}$, initialize dictionary $H[x']=0$\;
    \For{$i\gets 1$ to $|\D|$}{
        Find set $\mathcal{M}$ of all mappings between $D_i$ and $x$\;
        \For{$x'\in N^p_{fwd}$}{
            \If{$H[x']<f$ and $|\D|-i+H[x']\geq f$}{
                Let $(u,v)$ be the forward edge, i.e., $x'=x\diamond(u,v)$ and $v\in x',v\notin x$\;
                \For{$m\in \mathcal{M}$}{
                    \If{$\exists w\in V_{D_i}$ s.t. $(w,m(u))\in E_{D_i}$, $l(w)=l(v), w\notin m(V_{x})$}{
                    $H[x']\gets H[x']+1$\;
                    break\;
                    }}}}}
    \For{$x'\in N^p_{fwd}$}{
        \lIf{$H[x']\geq f$}{$N_1^p\gets N_1^p\cup \{x'\}$\;}
        \lElse{$N_2\gets N_2\cup \{x'\}$\;}
        }
    }
\Return $N_1^p, N_1^b, N_2$\;
\line(1,0){230}\\
\textbf{function  }{\FnSubFreq{$x',B_x$}{\\
    \Indp $B\gets \bigcap_{e\in x'}B_{e}$\;
    $C\gets \{i|i\in~B\backslash B_x, ~x'\subset D_i, D_i\in \D\}$\;
    \lIf{$|B_x|+|C|\geq f$}{\Return $true$\;}
    \lElse{\Return $false$\;}
}}
        
\end{algorithm}

\subsection{The EEN Algorithm}\label{sec:een-details}
A naive way to populate $N_1^b(x)$,$N_1^p(x)$ and $N_2(x)$ is to test each neighbor
of $x$ against the graph dataset $\D$. However, this is extremely inefficient
since $|N(x)|\cdot|\D|$ isomorphism tests are required, where $|\D|$ is the number
of graphs in $\D$.
A simple optimization would be using the monotonic property of frequent patterns:
if $x$ is a frequent pattern, any subgraph of $x$ should be frequent too; likewise,
an infrequent pattern's super-graph must be infrequent.
However, the naive method is still required for exploring $N^p(x)$ if $x$
is frequent or $N^b(x)$ if $x$ is infrequent. 

The EEN algorithm is able to further optimize the number of isomorphism tests.
Observing that $x$ and $y$ only differ in one edge for all $y\in N(x)$,
the main idea is to re-use the isomorphic mappings
between $x$ and $D_i \in \D$ and examine whether any of the isomorphic mappings
can be retained after extending an edge.
The EEN algorithm is formally presented in Algorithm 2 and is described
in the following.

Algorithm 2 takes pattern $x$, graph dataset $\D$ and support threshold 
$f$ as input and returns $N_1^{b}(x)$, $N_1^{p}(x)$ and $N_2(x)$.
First, pattern $x$ is tested against each graph in
$\D$ and the result is stored in $B_x=\{i|x\subset D_i, D_i \in \D\}$,
which is the set of IDs of graphs containing pattern $x$ (line 2).
The subgraph isomorphism algorithm we use is the VF2 algorithm \cite{cordella2004sub}.
Next we populate three types of neighbors of $x$:
sub-neighbors $N^b$, super-back neighbors $N_{back}^p$ and
super-forward neighbors $N_{fwd}^p$ (line 3), and handle them differently.

\para{Explore sub-neighbors (line 4 to 7).}
For $N^b$ , if $x$ is frequent, the entire set $N^b$ should
be frequent. If $x$ is infrequent, each pattern in $N^b$ is examined by 
the boolean sub-procedure {\sc Sub\_is\_freq} (line 40 to 44). 
{\sc Sub\_is\_freq} takes a sub-neighbor $x'$ of $x$ and $B_x$ as input
and returns the frequentness of $x'$.
First we find $B_E=\bigcap_{e\in x'}B_{e}$,
the intersection of ID sets of all edges in pattern $x'$. 
Then subgraph isomorphism
test is only needed for the graphs $D_i\in B_E\backslash B_x$.
The set $C$ of IDs of graphs that succeed the test together with $B_x$ comprise $B_{x'}$.
Finally the procedure returns the frequentness of $x'$
by comparing $f$ and the size of $B_{x'}$.

\para{Explore super-back neighbors (line 8 to 22).}
For $N_{back}^p$, 
if $x$ is infrequent, the entire $N_{back}^p$ must be infrequent.
Otherwise, 
we test whether $x'\in N_{back}^p$ is a subgraph of $D_i$ for each $D_i$.
In this part, the EEN algorithm does not require any additional
subgraph isomorphism test at all.
This is achieved by re-using the isomorphism mappings between the base pattern $x$
and $D_i$ and reasoning upon that.
In line 12 we find the subgraph isomorphism mappings 
$\mathcal{M}:V_x^n \rightarrow V_{D_i}^n$, which can be obtained at the same time
when computing $B_x$ in line 2.
Suppose $x$ is extended to $x'$ by connecting node $u$ and $v$ (line 15). 
If any of the isomorphism mappings $m\in \mathcal{M}$
is preserved with the edge extension (i.e., $m(u)$ and $m(v)$ are adjacent in $D_i$),
then $x'$ must be a subgraph of $D_i$. Otherwise if none of the mappings can be
preserved, $x'$ is not a subgraph of $D_i$.

In the above process, we use a dictionary $H$ to keep track of the number
of graphs in $\D$ so far which contains $x'$ as a subgraph, i.e., $H[x']$ maintains
$|\{D_i|x'\subset D_i\}|$ for the $D_i$ tested so far.
Line 14 ensures that the isomorphism extension
test is only performed when $H[x']$ has not and is able to reach $f$.

\para{Explore super-forward neighbors (line 23 to 37).}
For $N_{fwd}^p$, the algorithm is
similar to the procedures of exploring super-back neighbors, 
except that the extension test is now on a forward
edge instead of a back edge.
Specifically, let $v$ be the new node extended from $u$ (line 30),
if there exists a node $w\in D_i$ satisfying 1) has the same label as $u$;
2) is adjacent to $m(u)$; and 3) is not part of the mapping $m$,
then the isomorphism can be extended, meaning $x'\subset D_i$.

\section{Privacy and Utility Analysis}\label{sec:analysis}
\subsection{Privacy Analysis}\label{subsec:privacy}
In this part we establish the privacy guarantee of {\em Diff-FPM} described above.
We show both the sampling and perturbation phases preserve privacy, and then
we use the composition property of differential privacy to show the privacy guarantee of the overall algorithm.

In the sampling phase, our target probability distribution
$\pi(\D,\cdot)$ equals $\frac{\exp(\eps_1 u(\D,\cdot)/2k\Delta u)}{C}$ for a given dataset $\D$.
If samples were drawn directly from this distribution, it would
achieve strict $\frac{\eps_1}{k}$-differential privacy due to the exponential mechanism.
Since we use MCMC based sampling, the distribution of the samples
$\hat{\pi}(\D,\cdot)$
will {\em approximate} $\pi(\D,\cdot)$, i.e. the two distributions are
{\em asymptotically } identical. In real simulation, there may be a small distance between the two distributions. To quantify the impact on privacy when
a small error is present, we use the {\em total variation distance} \cite{rubinstein2008simulation}
to measure the distance of the two distributions at a given time:
\begin{equation}
\label{def:tv}
|| \hat{\pi}(\cdot)-\pi(\cdot)||_{TV} \equiv \max_{T\subset\X}|\hat{\pi}(T)-\pi(T)|
\end{equation}
which is the largest possible difference between the probabilities that
$\pi(\cdot)$ and $\hat{\pi}(\cdot)$ can assign to the same event.

Let $\mathcal{A(\D)}$ denote the process of sampling one pattern
according to Algorithm 1 (Line 4 to 10).
The privacy guarantee that $\mathcal{A}(\D)$ offers is described by the following lemma:
\begin{lemma}
\label{thm:error}
Let $\pi(\cdot)$ and $\hat{\pi}(\cdot)$ denote the target distribution and
the distribution of samples from $\mathcal{A}(\D)$ respectively.
Suppose $|| \hat{\pi}(\cdot)-\pi(\cdot)||_{TV} \leq \theta$, procedure
$\mathcal{A}(\D)$ gives
$(\frac{\eps_1}{k},\delta)$-differential privacy, where $\delta=\theta(1+e^{\eps_1/k})$.
\end{lemma}
\begin{proof}
$\forall x \in \X$, the ratio of density at $x$ for two neighboring input
$\D$ and $\D'$ can be bounded as
\begin{align}
\frac{\hat{\pi}(\D,x)}{\hat{\pi}(\D',x)} &\leq \frac{\pi(\D,x)+\theta}{\hat{\pi}(\D',x)}\nonumber\\ 
&\leq \frac{\pi(\D',x)\cdot e^{\eps_1/k}+\theta}{\hat{\pi}(\D',x)}\nonumber\\
&\leq \frac{\Big(\theta+\hat{\pi}(\D',x)\Big)e^{\eps_1/k}+\theta}{\hat{\pi}(\D',x)}\nonumber\\
&= e^{\eps_1/k} + \frac{\theta(1+e^{\eps_1/k})}{\hat{\pi}(\D',x)}\nonumber
\end{align}
Therefore,
\begin{equation}
\hat{\pi}(\D,x)\leq e^{\eps_1/k} \hat{\pi}(\D',x) + \theta(1+e^{\eps_1/k})\nonumber
\end{equation}
giving $\Big(\frac{\eps_1}{k},\theta(1+e^{\eps_1/k})\Big)$-differential privacy.
\end{proof}

Note that $\theta$ is a function of simulation time $t$. The following lemma
describes the asymptotic behavior and the speed of convergence of the chain :
\begin{lemma}\label{thm:asympto}
\cite{rubinstein2008simulation}
If a Markov chain on a finite state space is irreducible and aperiodic, and has
a transition kernel $P$ and stationary distribution $\pi(\cdot)$, then for
$x\in \X$,
\begin{equation}
||P^t(x,\cdot) - \pi(\cdot)||_{TV}  \leq M\rho^t, \qquad t = 1,2,3,\dots
\end{equation}
for some $\rho<1$ and $M<\infty$. And
\begin{equation}
\lim_{t\to\infty}||P^t(x,\cdot) - \pi(\cdot)||_{TV} = 0
\end{equation}
\end{lemma}

The theorem above means $\theta$ is decreasing at least at a geometric speed
and approximates to zero when the simulation is running long enough.

Since the sampling process in Algorithm 1 consists of $k$ successive
applications of exponential mechanism based on random walk,
we need the following well-known composition lemma to provide privacy guarantee for the
entire sampling phase.
\begin{lemma}~\cite{mcsherry2009differentially}\label{thm:epscomp}
Let $\A_1,\ldots,\A_t$ be $t$ algorithms such that $\A_i$ satisfies
$\eps_i$-differential privacy, $1\leq i\leq t.$ Then their
sequential composition $\langle\A_1,\ldots,\A_t\rangle$ satisfies
$\eps$-differential privacy, for
$\eps = \sum_{i=1}^t\eps_i$.
\end{lemma}

Equipped with the results in previous lemmas, we are able to provide
the privacy guarantee for Algorithm 1.
\begin{theorem}\label{thm:main-privacy}
Algorithm 1 satisfies $\eps$-differential privacy.
\end{theorem}
\begin{proof}
According to Lemma~\ref{thm:asympto}, when the chain has reached the
steady state, $\theta$ in Lemma~\ref{thm:error} becomes zero, giving
$\frac{\eps_1}{k}$-differential privacy in each output pattern. Using the
composition lemma, the sample phase satisfies $\eps_1$-differential privacy
as a whole. In the perturbation step, we add Laplace noise $Lap(k/\eps_2)$
independently on each of the true supports of the $k$ patterns. Again by
Lemma~\ref{thm:epscomp}, the perturbation phase gives $\eps_2$-differential
privacy. Therefore the entire Algorithm 1 achieves $\eps$-differential privacy
since $\eps=\eps_1+\eps_2$.
\end{proof}

\subsection{Utility Analysis}\label{subsec:utility}
Because neighboring inputs must have similar output under differential privacy,
a private algorithm usually does not return the exact answers. In the scenario of mining
top-$k$ frequent patterns, the {\em Diff-FPM} algorithm should return a noisy list of patterns
which is close to the real top-$k$ patterns. To quantify the quality of the output of {\em Diff-FPM},
we first define two utility parameters, following \cite{Thakurta2010}.
Recall that $f$ is the support of the $k$th frequent pattern, and let $\beta$
be an additive error to $f$. Given $0<\gamma<1$, we require that with
probability at least $1-\gamma$, (1) no pattern in the output has true support
less than $f-\beta$ and (2) all patterns having support greater than $f+\beta$
exist in the output.
The following theorems provide the utility guarantee of {\em Diff-FPM}.
A score function $u(x)=|gid(x)|$ is assumed.

\begin{theorem}
At the end of the sampling phase in Algorithm 1, for all $0<\gamma<1$,
with probability at least $1-\gamma$,
all patterns in set $S$ have support greater than $f-\beta$, where $\beta=
\frac{2k}{\eps_1}(\ln(k/\gamma)+\ln M)$ and $M$ is an upper bound on the size
of output space.
\end{theorem}
\begin{proof}
In any of the $k$ rounds of sampling, the probability of choosing a pattern
with support $f-\beta$ given that a pattern having support $\geq f$ is still
present is at most $e^{\frac{\eps_1(f-\beta)}{2k}}/e^{\frac{\eps_1 f}{2k}}
=\exp(-\eps_1 \beta/2k)$. Although the size $m$ of the output space is unknown
without enumeration, one can usually get an upper bound $M$ without considering
the isomorphism classes. Since there are at most $M$ patterns with support less
than $f-\beta$, after $k$ rounds of sampling the probability is upper bounded
by $kM\exp(-\eps_1 \beta/2k)$. Then
\begin{align}
\gamma &\geq kM\exp(-\eps_1 \beta/2k)\nonumber\\
\Leftrightarrow \beta &\geq \frac{2k}{\eps_1}\ln (kM/\gamma) \nonumber
\end{align}
\end{proof}
The following theorem provides the upper bound of noise added to the true
support of each output pattern.
\begin{theorem}
For all $0<\gamma<1$, with probability of at least $1-\gamma$, the noisy
support of a pattern differs by at most $\beta$, where $\beta=\frac{k}{\eps_2}
\ln(1/\gamma)$.
\end{theorem}
\begin{proof}
This is a property directly followed by integrating the Laplace distribution:
$\gamma \leq 2\int_{\beta}^{\infty}\frac{\eps_2}{2k}\exp(\frac{-\tau \eps_2}{k})d\tau=
\exp(\frac{-\beta\eps_2}{k})$, which transforms to $\beta\leq\frac{k}{\eps_2}
\ln(1/\gamma)$.
\end{proof}

\section{Experimental Study}\label{sec:exp}

In this section, we evaluate the performance
of {\em Diff-FPM} through extensive experiments on various datasets. 
Since this is the first work on differetially private mining
of frequent graph patterns, the quality of the output is compared with
the result from a non-private FPM algorithm and the accuracy is reported.
In addition, we demonstrate the effectiveness of the EEN algorithm by
comparing the time cost per iteration to two basic methods.
We also discuss the running time and scalability of {\em Diff-FPM}
and the impact
of various parameters 
such as privacy budget, the number of output patterns and the size of
the graph dataset.
In this section
we consider the scenario of mining the top-$k$ frequent patterns.

\subsection{Experiment Setup}

\para{Datasets.}
The following three datasets are used in our experiment:
{\em DTP} is a real dataset containing 
DTP AIDS antiviral screening
dataset\footnote{\footnotesize \url{http://dtp.nci.nih.gov/docs/aids/aids_data.html}},
which is frequently used in frequent graph pattern mining study.
It contains 1084 graphs, with an
average graph size of 45 edges and 43 vertices. There are 14 unique node
labels and all edges are considered having the same label.

The {\em click} dataset
consists of 20K small tree graphs (4 nodes and 3 edges on average) obtained by a graph generator
developed by Zaki \cite{zaki2005efficiently}. To a certain extent, this
synthetic dataset simulates user click graphs from web server
logs \cite{zaki2005efficiently}, which
is a suitable type of data requiring privacy-preserving mining.
All the tree graphs in this dataset are sampled from a master tree.
In our experiment the master tree has 10,000 nodes
with a depth of 10 and a fanout of 6.

The above two datasets contain graphs that are relatively sparse.
To test our algorithm on {\em dense} graphs, 
we also use a dataset 
containing 5K graphs, in which
the average node degree is 7. Each graph contains 10 vertices and 35 edges
on average.
The graph generator \cite{graphgen} we use is specially designed for generating
graph datasets for evaluation of 
frequent subgraph mining algorithms. The size of this graph dataset
is comparable to the largest datasets used in previous works \cite{yan2002gspan,inokuchi2000apriori}.

\para{Utiliy metrics.}
We evaluate the quality of the output of {\em Diff-FPM} by employing the
following two utility metrics: 
\begin{itemize}
\item {\em Precision.}
Precision is defined as the fraction of 
identified top $k$ graph patterns that are in the actual top $k$,
i.e., 
\[Precision = \frac{|\textrm{True Positives}|}{k}\]
This is the complementary measure
of the false negative rate used in \cite{Thakurta2010}.
\item{\em Support Accuracy.}
The measure of precision reflects the percentage of desired/undesired
patterns in the output, yet it cannot indicate how good or bad the output patterns are in terms of their supports.
For example, if $f=1000$, it is much more undesirable
if a pattern with support 10 appears in the output compared to a pattern with support
980, even though the precision may be the same in these two cases.
We first define the {\em relative support error (RSE)} as
\[RSE = \frac{(S_{true}-S_{out})/k}{f}\]
where $S_{true}$ and $S_{out}$ are the sum of the supports of the real top-$k$ patterns and sum of the supports of the sampled patterns respectively. This measure reflects
the average deviation of an output pattern's support with respect to
the support threshold $f$.
In the plots, the {\em support accuracy} is reported, which equals
$1-\textit{RSE}$.
\end{itemize}


All experiments were conducted on a PC with 3.40GHz CPU with 8GB RAM.
The random walk in the {\em Diff-FPM} algorithm consumes 
only a small amount of memory 
due to its Markovian nature, i.e.,
earlier states in the walk do not need to be remembered.
We can, however, allocate extra memory to cache some of the patterns and their neighbors.
We implemented our algorithm in Python 2.7 with the JIT compiler PyPy\footnote{\url{http://pypy.org}}
to speed up. The default parameters of $\eps=0.5$ and $k=15$ were used unless
specified otherwise.
In the experiment we do not release the noisy supports of the patterns in
the output (line 9 in Algorithm 1), so all the privacy budget is used in
the sampling phase.

\subsection{Experiment Results}
\para{Comparison of neighbor exploration methods.}
In Section~\ref{sec:een} we proposed the EEN algorithm to efficiently explore
the neighborhood of a pattern. We now compare it with two other methods:
a {\em naive} approach which finds the support of each neighbor of the current pattern $x$
and a {\em basic} approach which uses the monotonic property of frequent patterns 
(see Section~\ref{sec:een-details}). Figure~\ref{fig:een} shows the average iteration time
in logarithm of the three methods over three datasets. In each iteration, a neighboring pattern is proposed
and then accepted or rejected according to the MH algorithm. Clearly, EEN takes significantly
less time in each iteration than the other methods in both datasets, reducing the iteration
time by at least an order of magnitude compared to the naive approach. Thus all
subsequent results are presented with EEN enabled.

\begin{figure}[t]
\centering
\includegraphics[width=0.35\columnwidth]{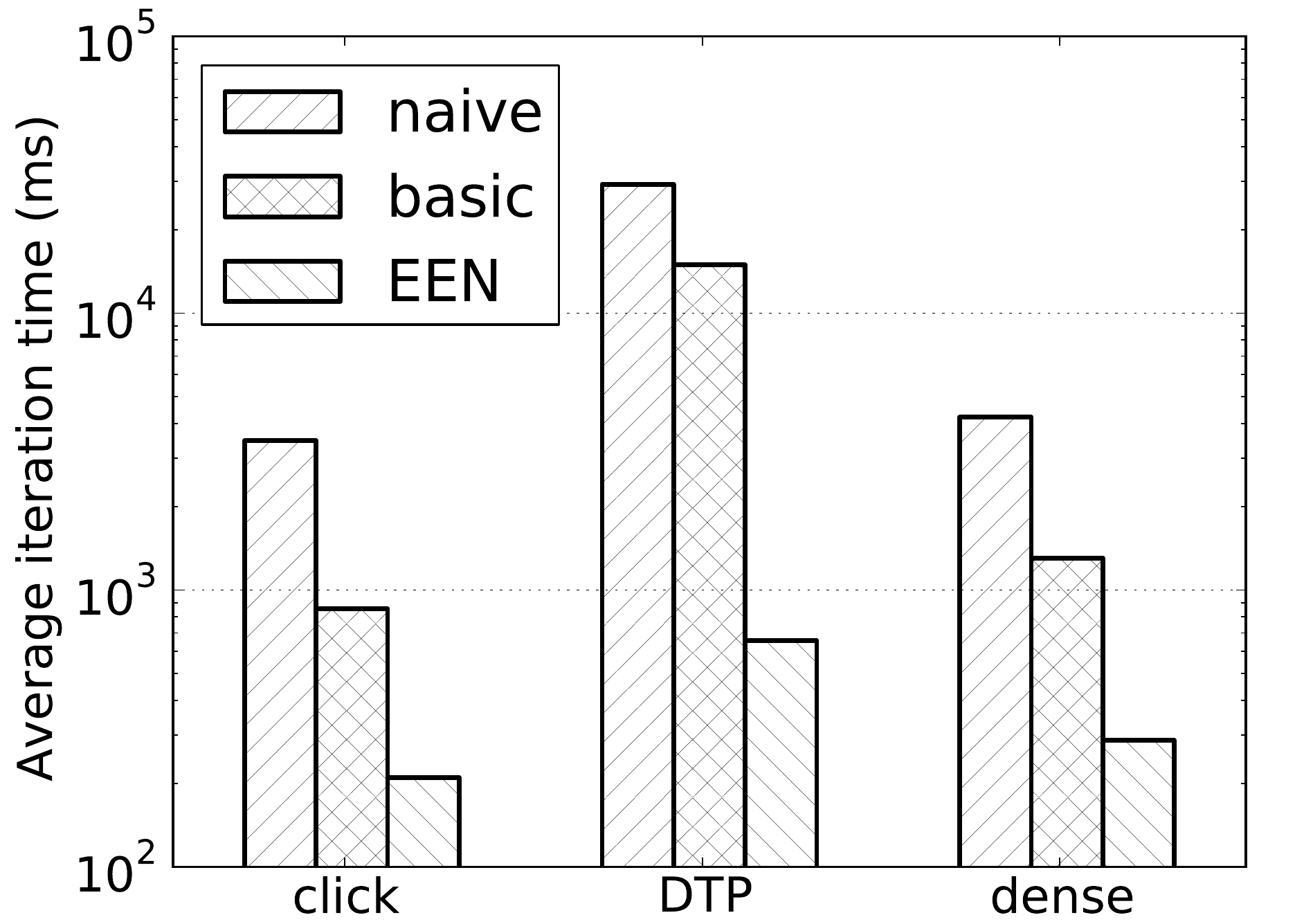}
\caption{Comparison of neighbor exploration methods}
\label{fig:een}
\end{figure}

\newlength{\figwidth}
\setlength{\figwidth}{0.33\textwidth}

\begin{figure*}[t]
\centering
\subfigure[]{\label{fig:n_pre}
\includegraphics[width=\figwidth]{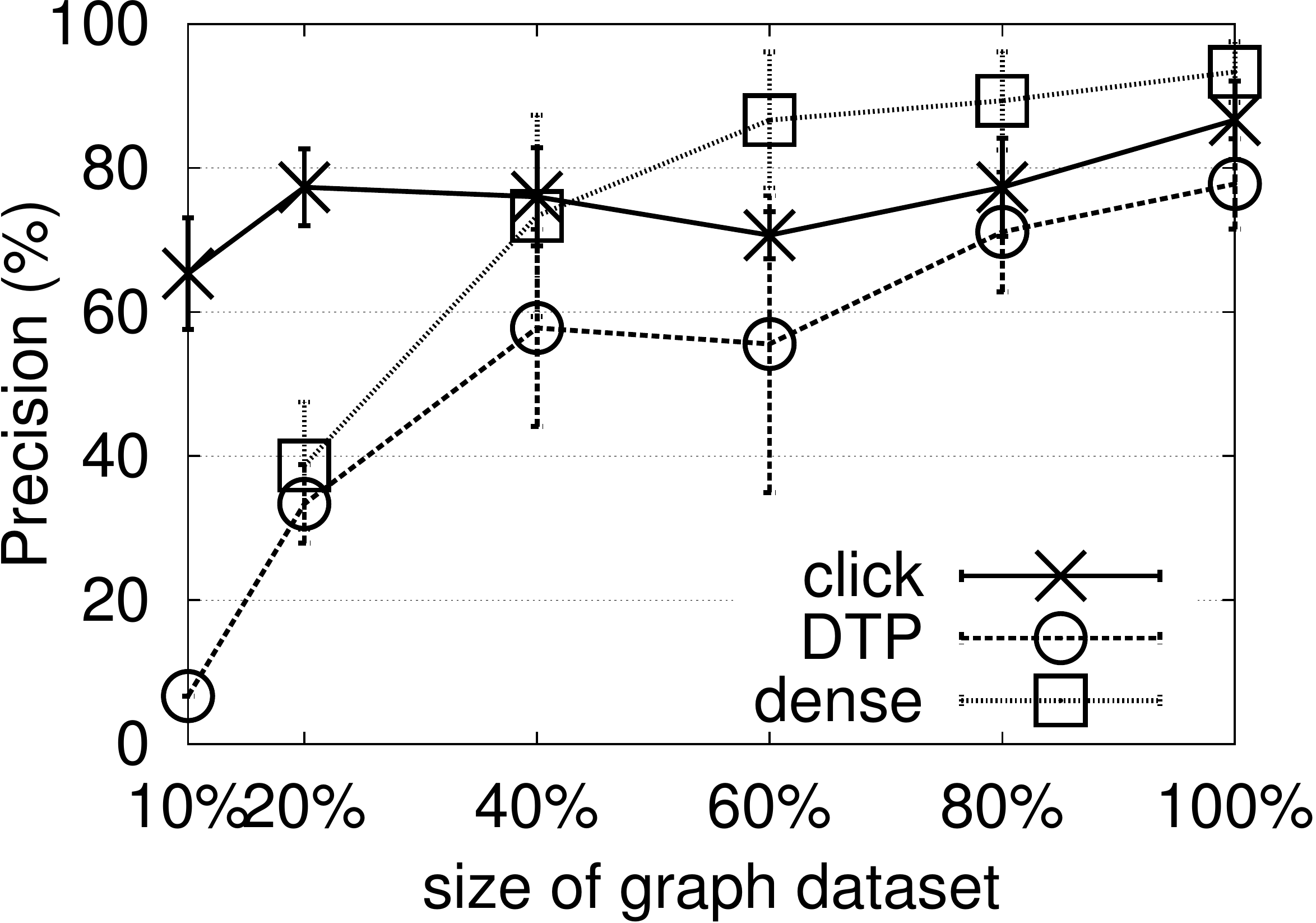}
}%
\subfigure[]{\label{fig:n_err}
\includegraphics[width=\figwidth]{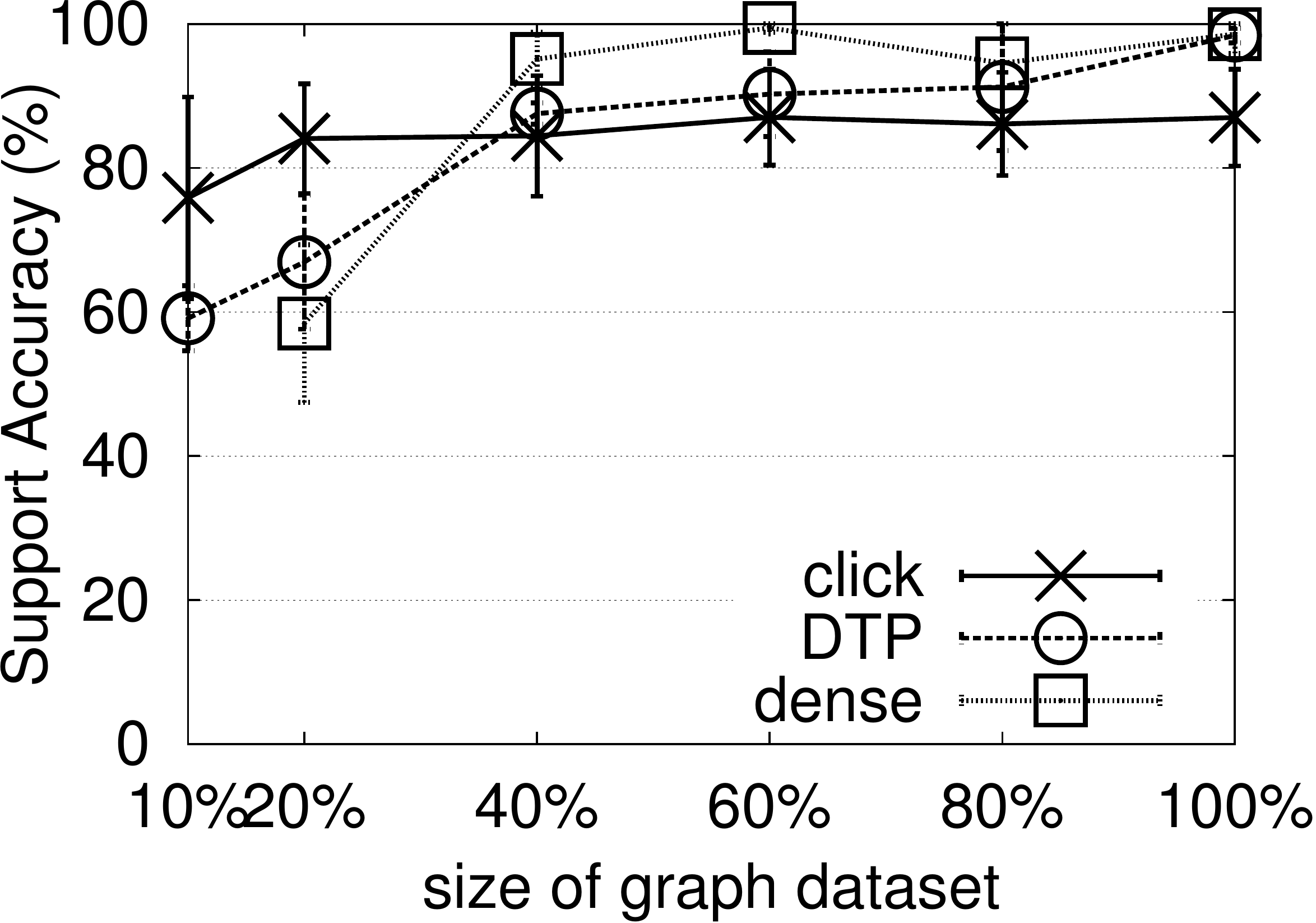}
}%
\subfigure[]{\label{fig:n_time}
\includegraphics[width=\figwidth]{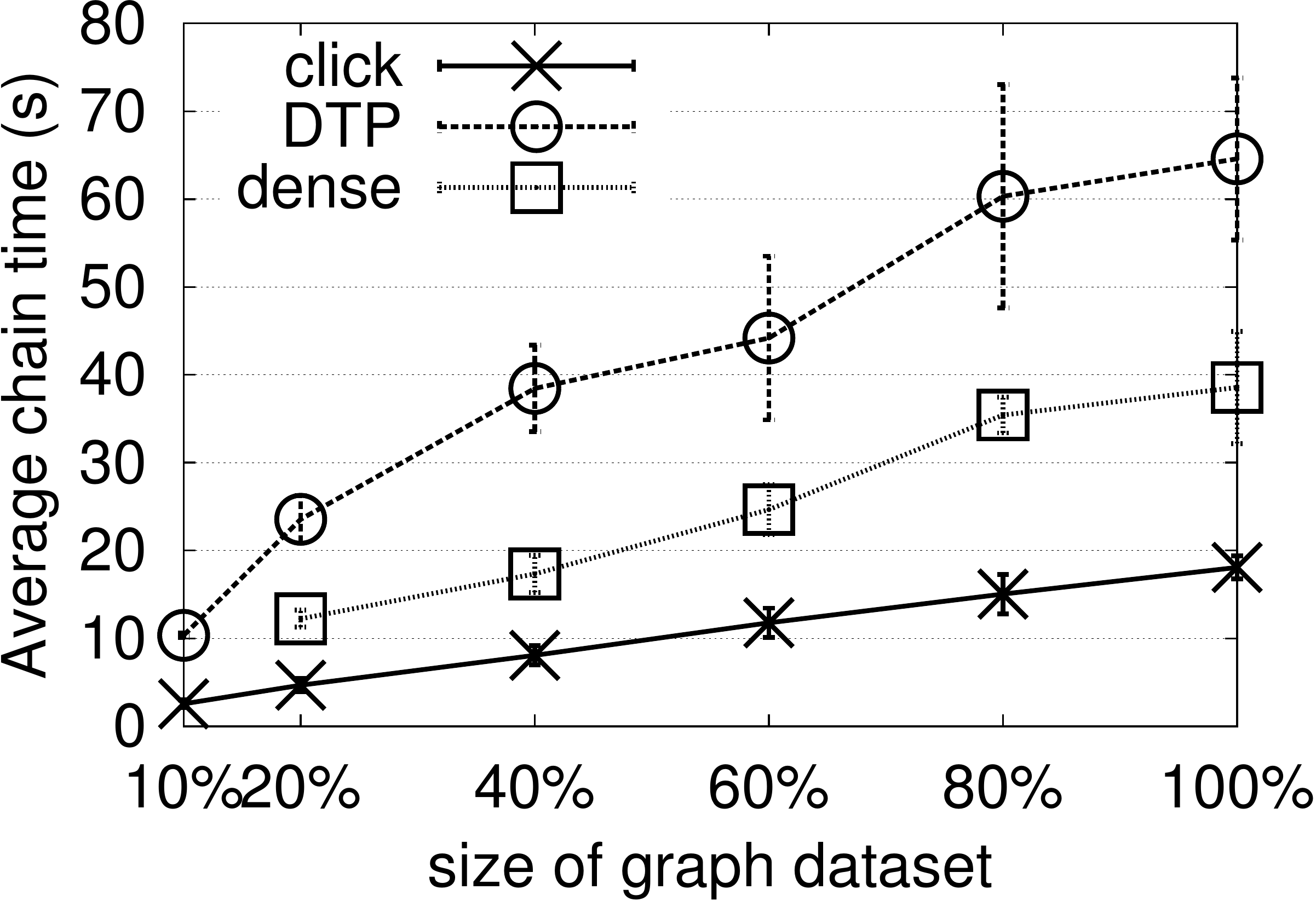}
}%
\caption{Impact of graph dataset size}
\label{fig:size}
\end{figure*}

\begin{figure*}[t]
\centering
\subfigure[]{\label{fig:score_pre}
\includegraphics[width=\figwidth]{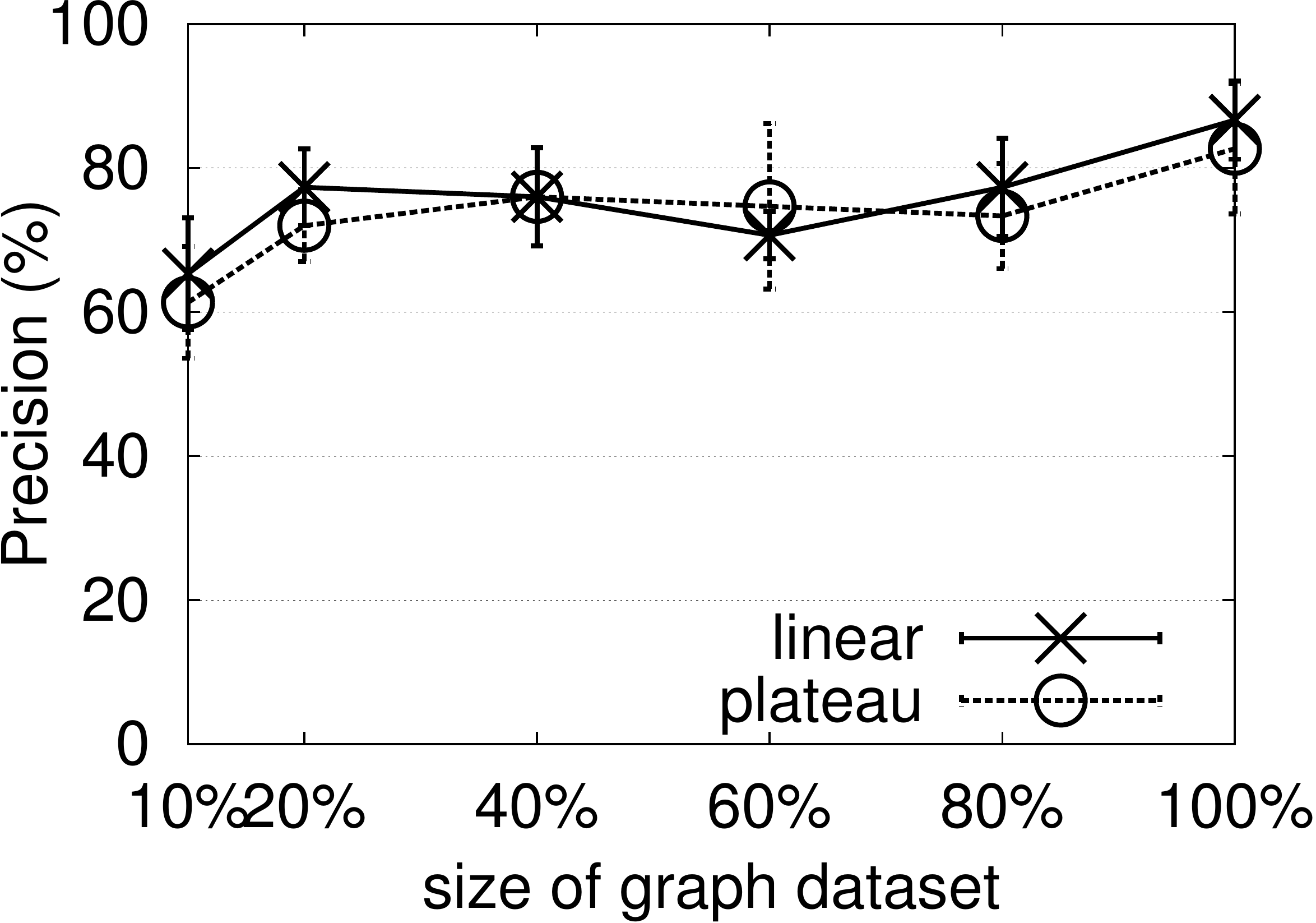}
}%
\subfigure[]{\label{fig:score_err}
\includegraphics[width=\figwidth]{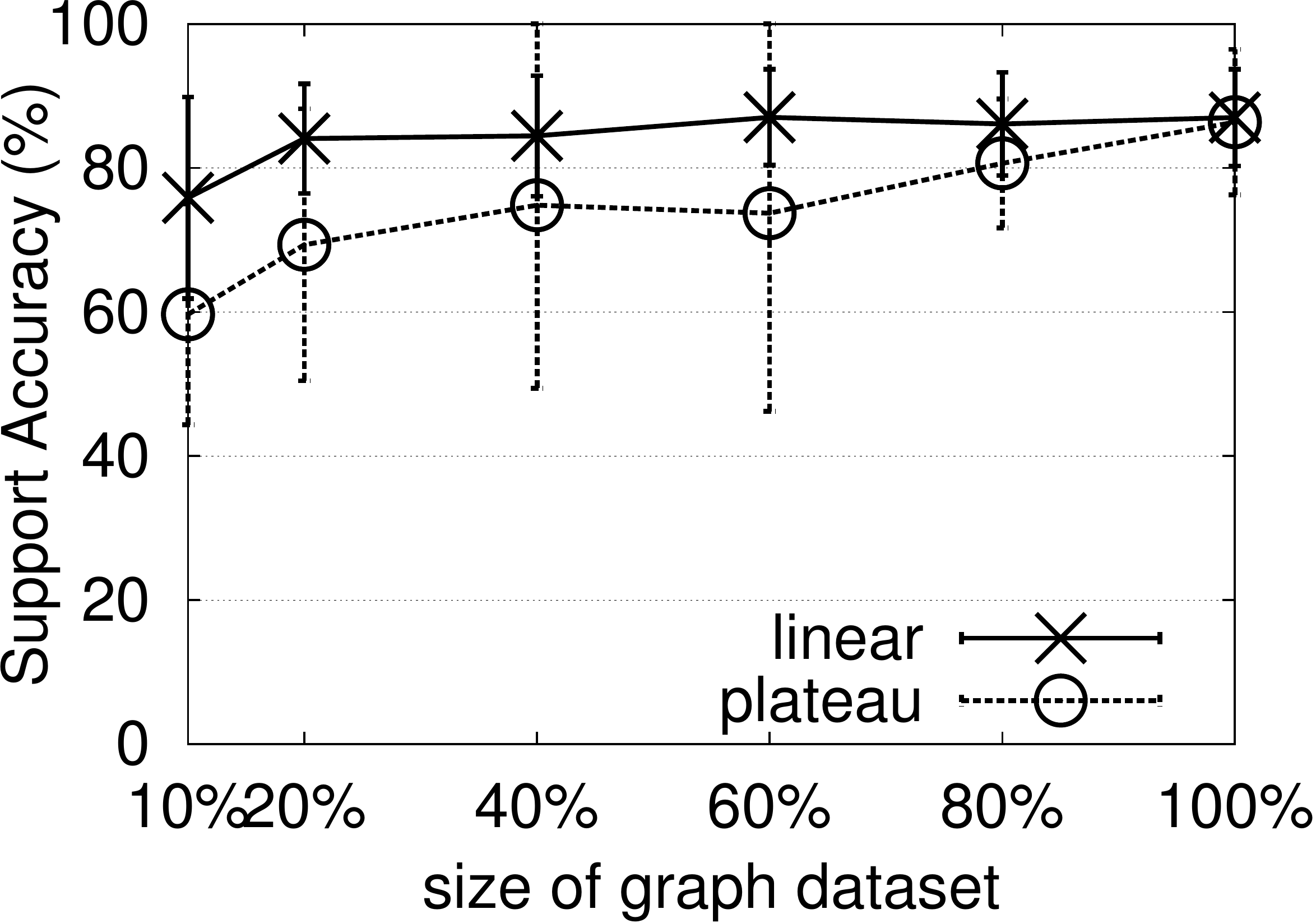}
}%
\caption{Score function}
\label{fig:score}
\end{figure*}

\begin{figure*}[t]
\centering
\subfigure[]{\label{fig:eps_pre}
\includegraphics[width=\figwidth]{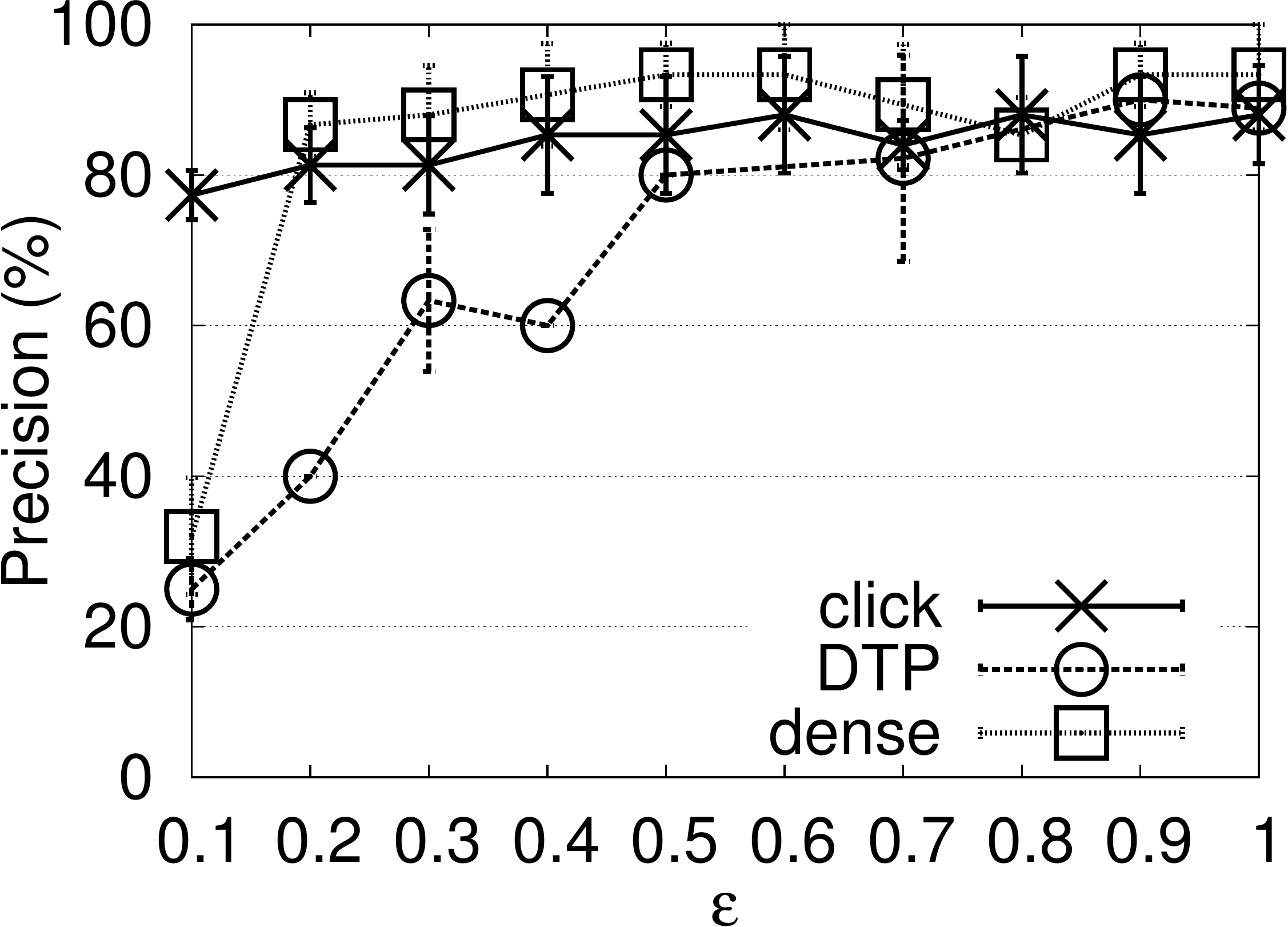}
}%
\subfigure[]{\label{fig:eps_err}
\includegraphics[width=\figwidth]{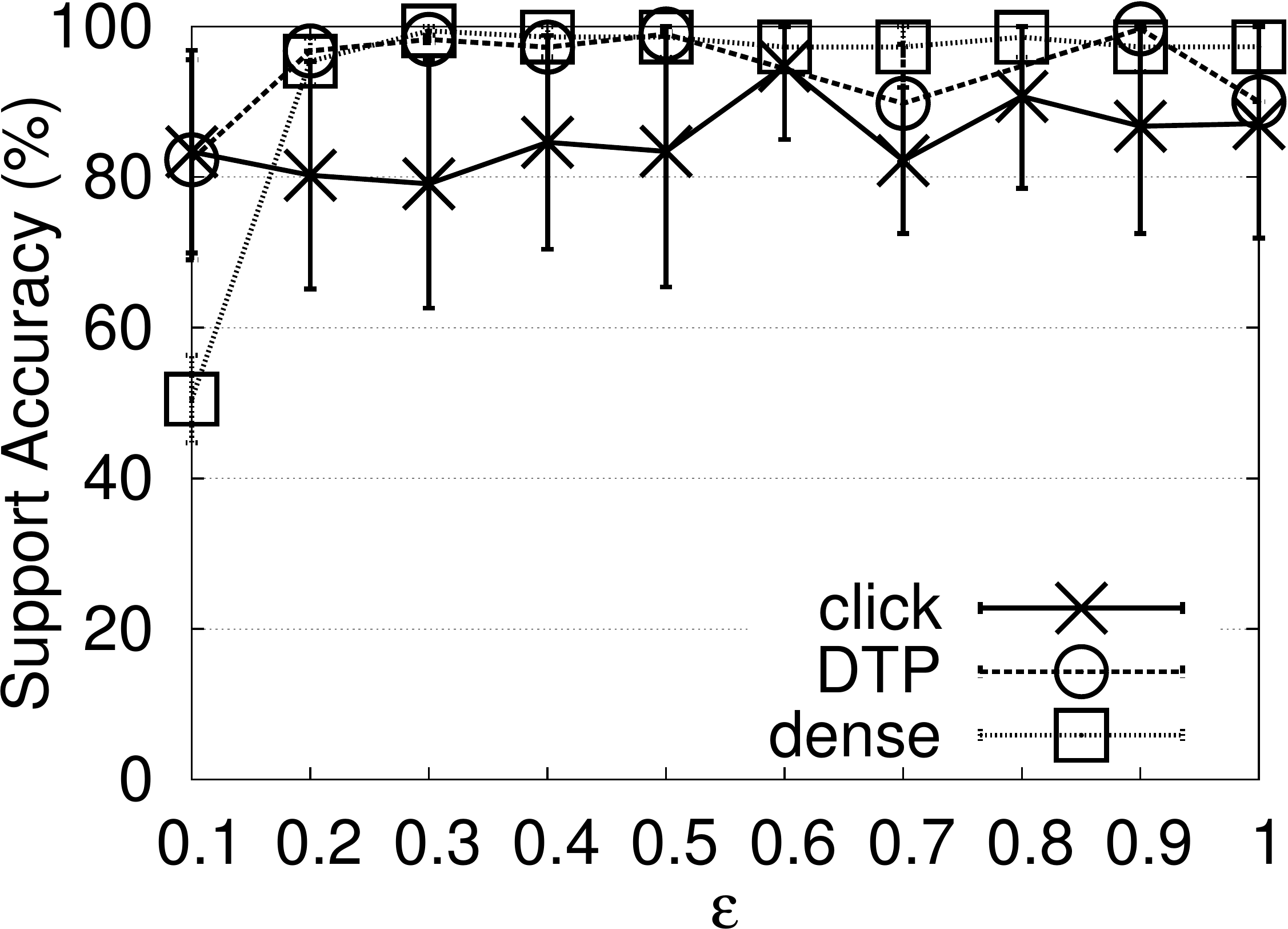}
}%
\caption{Precision and accuracy versus $\eps$}
\label{fig:eps}
\end{figure*}

\begin{figure*}[t]
\centering
\subfigure[]{\label{fig:k_pre}
\includegraphics[width=\figwidth]{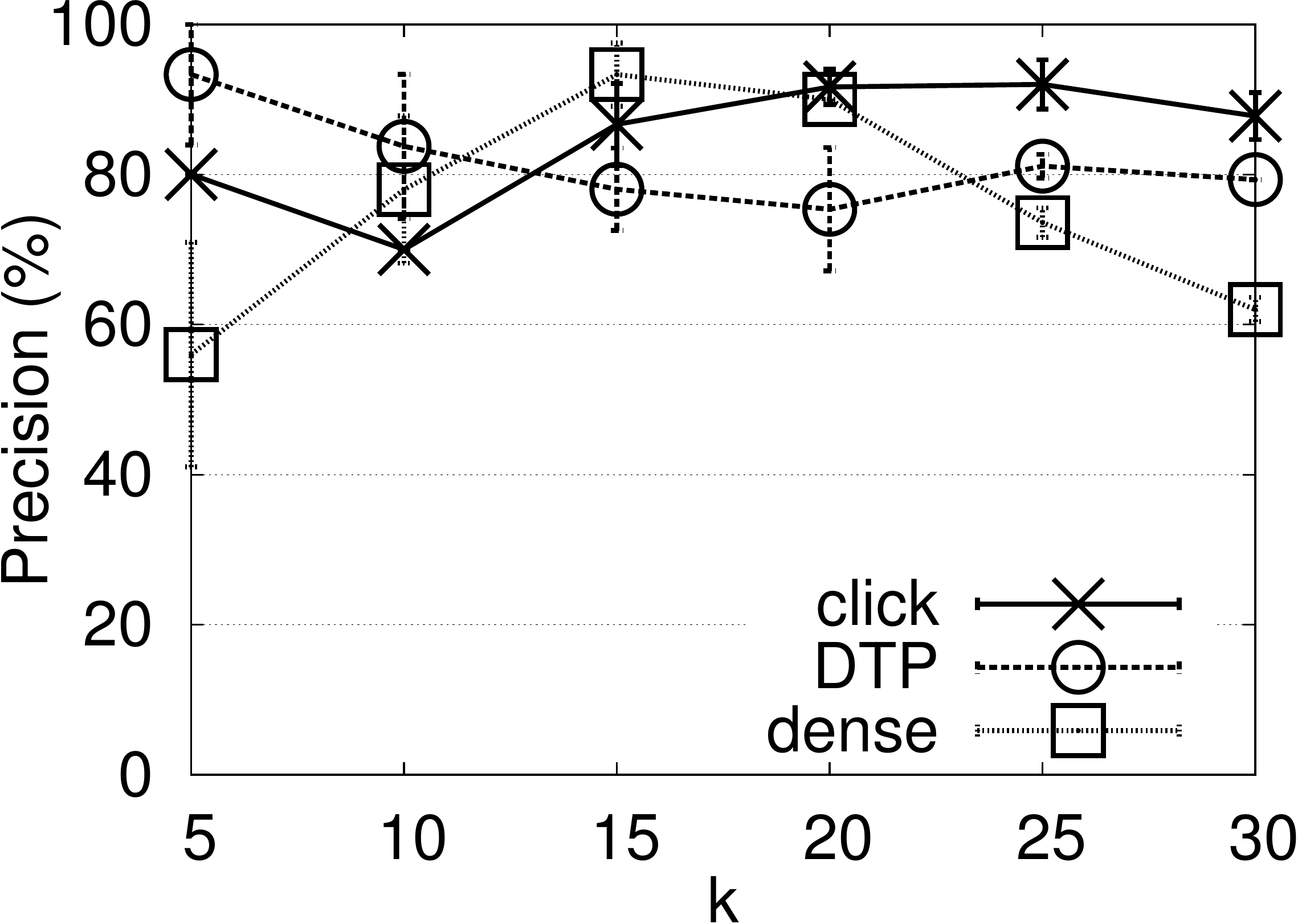}
}%
\subfigure[]{\label{fig:k_err}
\includegraphics[width=\figwidth]{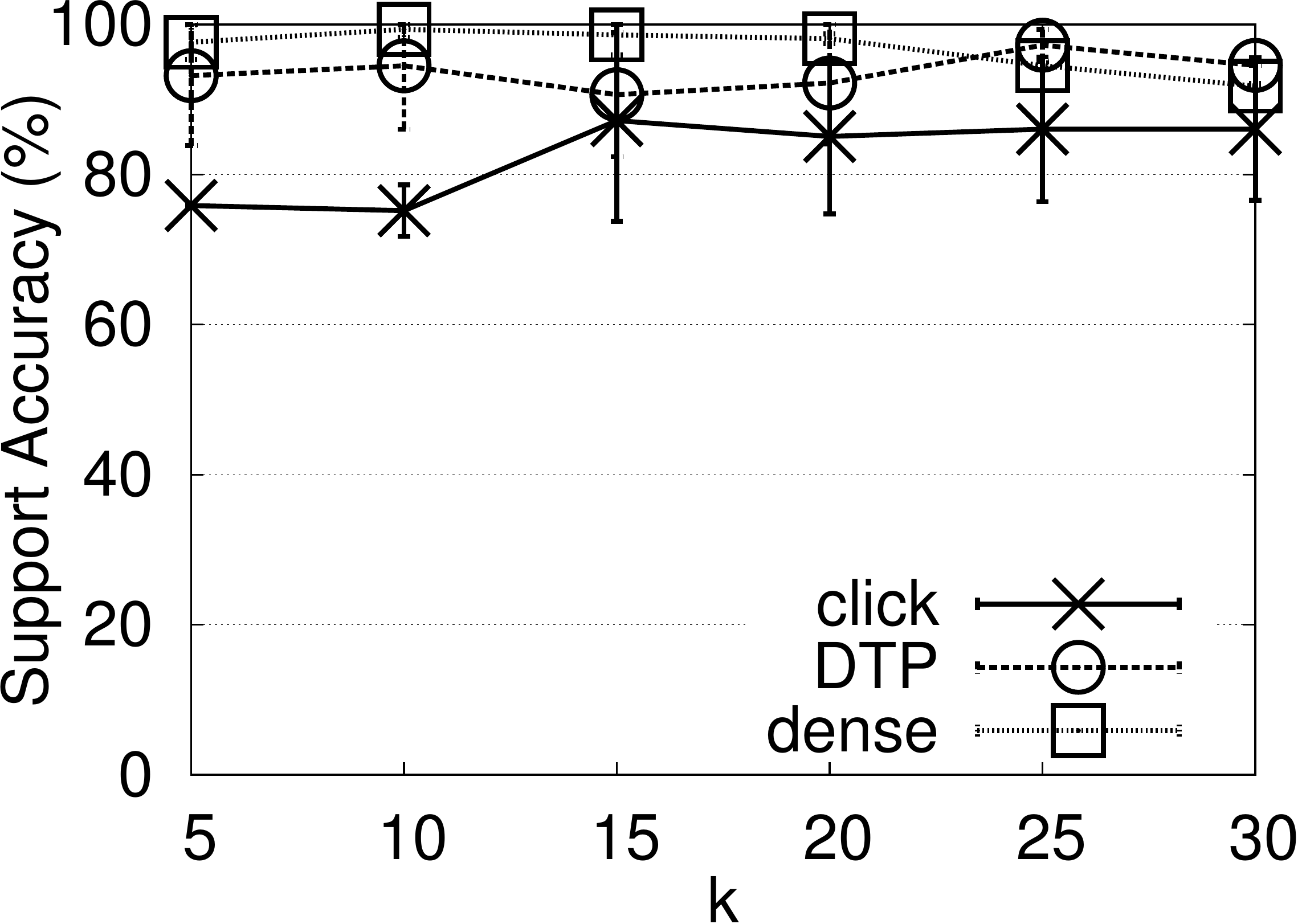}
}%
\caption{Precision and accuracy versus $k$}
\label{fig:k}
\end{figure*}

\para{Run time and scalability.}
Figure~\ref{fig:n_time} illustrates the average time taken to output one frequent pattern
as the size of the dataset increases. For the full datasets, {\em click} takes 20 seconds,
{\em DTP} takes about 1 minute and {\em dense} sits in the middle, although the {\em click} dataset contains 20K
graphs compared to only 1K in the {\em DTP}. It indicates that the size of each individual
graph and the size of the neighborhood have a larger impact on the run time than the total number of graphs in the dataset (note that {\em DTP} has 14 labels and thus a larger neighborhood of a pattern compared to {\em dense}).
For scalability, all datasets are observed to have linear scale-up in time as the size of graph dataset
increases.

\para{Precision and support accuracy.}
We now examine the quality of the output by studying the precision and support accuracy ({\em SA})
of the {\em Diff-FPM} algorithm under various parameter settings.

First, Figure~\ref{fig:n_pre} and Figure~\ref{fig:n_err} show the precision and {\em SA} when
we increase the size of the graph dataset from 10\% to 100\%
\footnote{\footnotesize The data point for {\em dense} at 10\% is absent since the smallest dataset size can be generated is 1K.}.
An increasing trend of the output quality can be clearly observed here. This is in line with
our expectation because achieving differential privacy is more demanding in a small dataset
-- the larger the number of records in the database, the easier it is to hide an individual record's
impact on the output. For all three full datasets, {\em Diff-FPM} is able to achieve at least
80\% on both precision and {\em SA}.

Figure~\ref{fig:eps_pre} shows the precision when varying privacy budget $\eps$.
With a very limited budget ($\eps=0.1$), only about 30\% of samples are
from the real top-$k$ patterns for {\em DTP} and {\em dense}. This is inevitable due to
the privacy-utility tradeoff.
As more privacy budget is given, the precision
of {\em Diff-FPM} increases fast. At $\eps=0.5$, the precisions from all
datasets have reached 80\%.
Further increase in privacy budget does not provide significant benefit on the
precision. We observed a similar trend in the support accuracy plot (Figure~\ref{fig:eps_err}),
with less dramatic changes for $\eps$ from 0.1 to 0.5.

Figures~\ref{fig:k_pre} and \ref{fig:k_err} illustrate the impact of the
number of patterns in the output.
Recall that in each round of sampling, a budget of $\eps/k$ is consumed
(cf. proof of Theorem~\ref{thm:main-privacy}).
Given a certain privacy budget, the more patterns
to output, the less privacy budget each sample can use. 
Thus we expect the
average quality of the output to drop as $k$ increases,
which is confirmed in the result.
Meanwhile, the support accuracy of the output holds well with the increasing number
of output, which can be seen in Figure~\ref{fig:k_err}.

\begin{figure}[t]
\centering
\subfigure[Accept rate vs $\eta$]{\label{fig:eta}
\includegraphics[width=0.35\columnwidth]{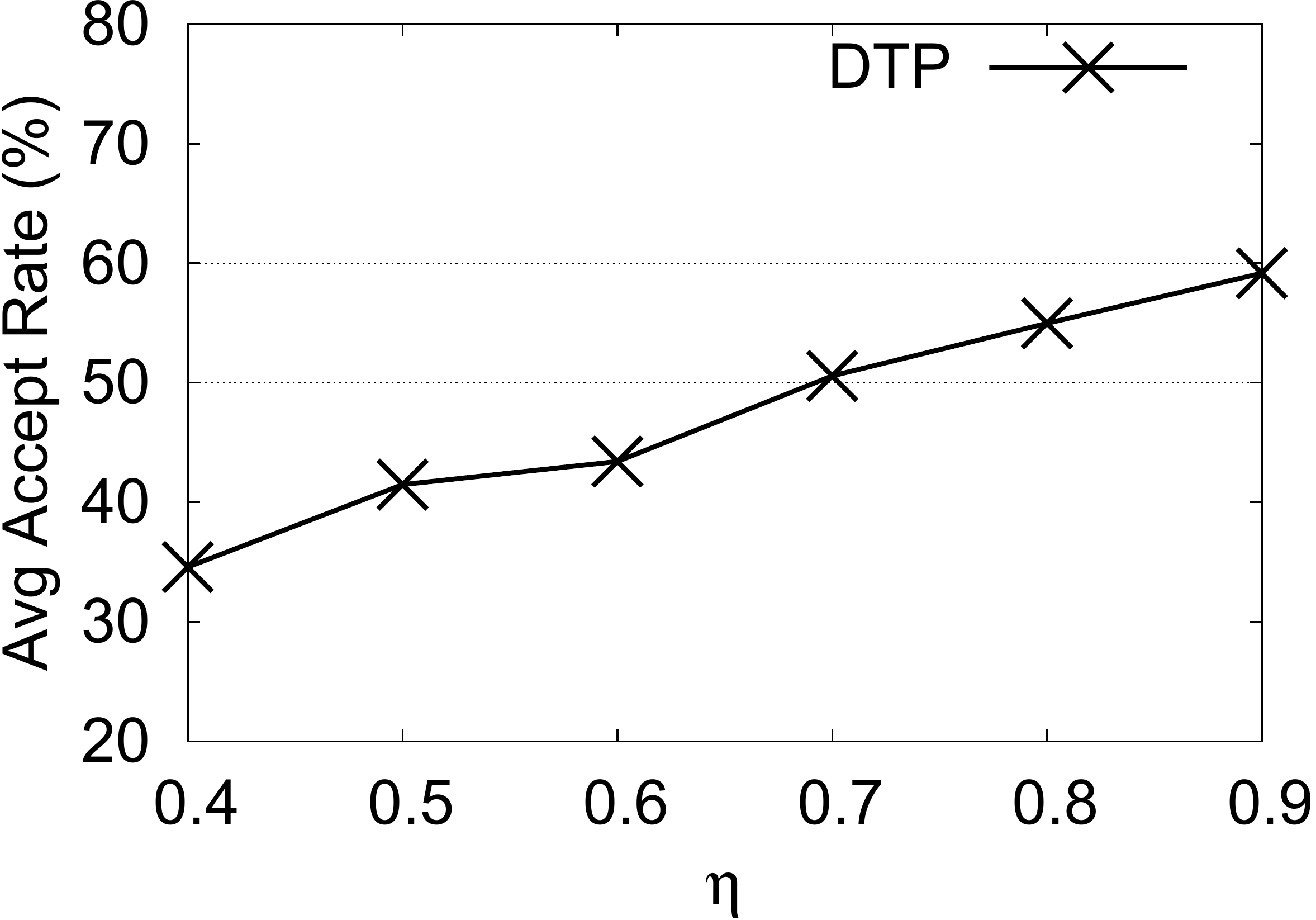}
}%
\subfigure[Accept rate vs $\rho$]{\label{fig:gamma}
\includegraphics[width=0.35\columnwidth]{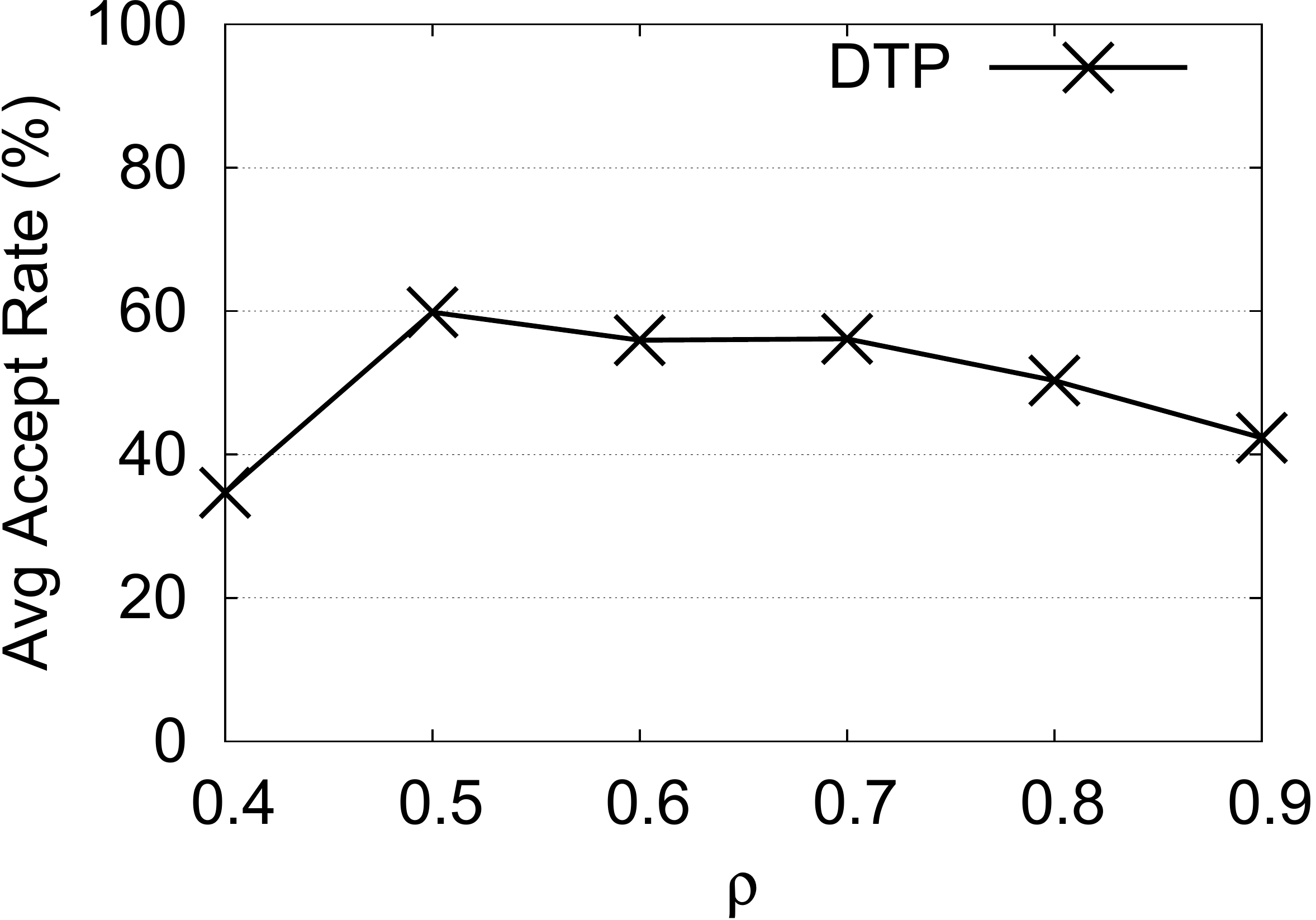}
}%
\caption{Impact of $\eta$ and $\rho$ on accept rate}
\label{fig:accept}
\end{figure}
\begin{figure}[t]
\centering
\includegraphics[width=0.5\columnwidth]{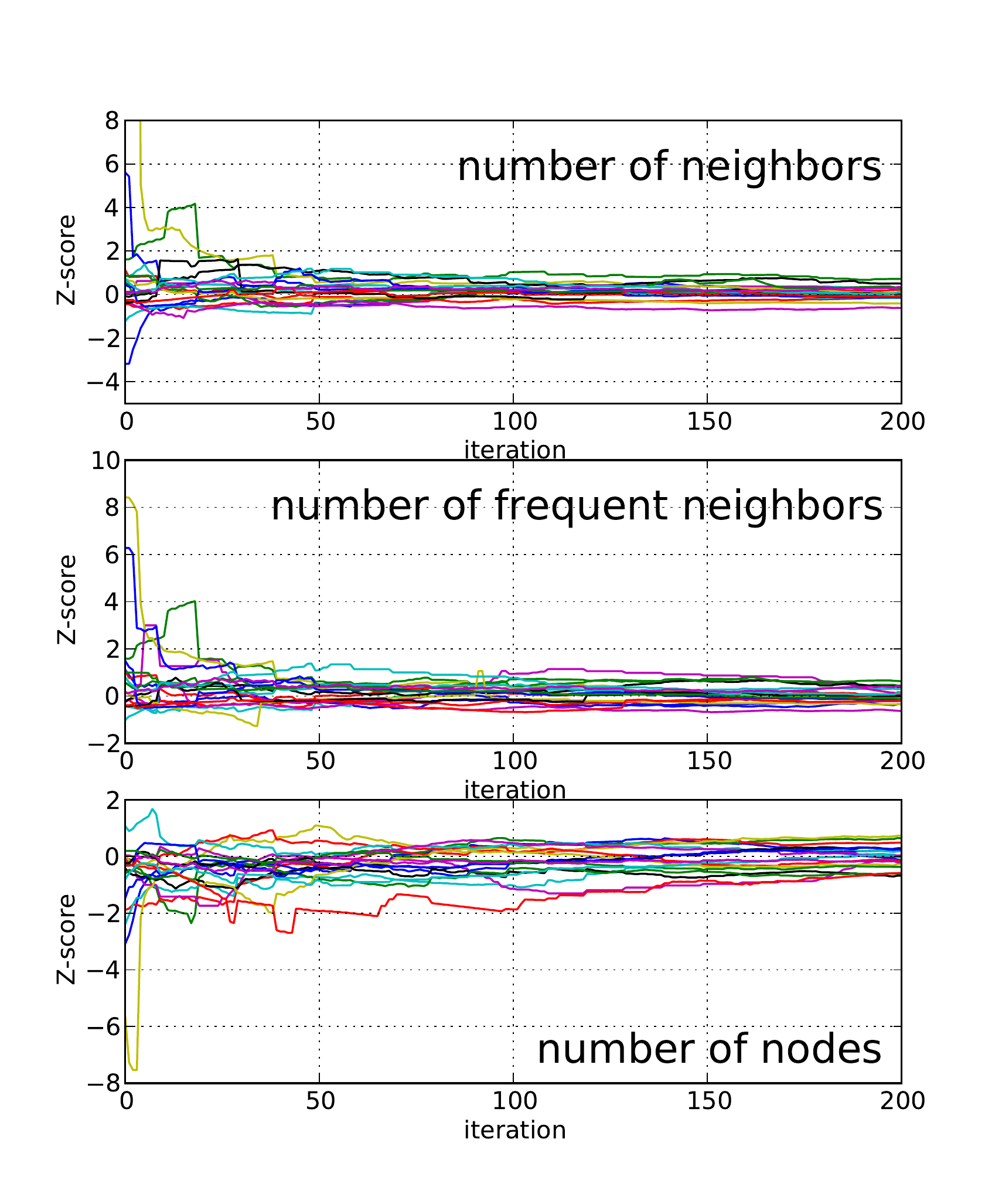}
\caption{Convergence trace of 20 chains}
\label{fig:converge}
\end{figure}

\para{Score function.}
In Section~\ref{sec:overview} we discussed the principles of designing the
score function. Here we experimentally compare several basic choices
on the synthetic dataset.
Figures \ref{fig:score_pre} and \ref{fig:score_err} show the precision and
support accuracy of two score functions {\em linear} and {\em plateau}.
{\em linear} represents the most straightforward choice:
$u(x)=|gid(x)|$ for any pattern $x$, with $\Delta u=1$.
{\em plateau} treats all the patterns in $\{x|gid(x)\geq f\}$ the same,
i.e., $u(x)=f$, if $|gid(x)|\geq f$; $u(x)=|gid(x)|$, if $|gid(x)|<f$.
The random walk with the {\em plateau} score function is able to traverse more
patterns in the POFG. However, as shown in the plots, this does not lead to
better precision and support accuracy in the result. Over the range of
different graph dataset sizes, the {\em linear} score function consistently
performs better due to the exponentially amplified probability mass for
more frequent patterns. Therefore we use the {\em linear} score function
for the rest of the experiment.

\para{Impact of proposal distribution.} Recall that two parameters have
impact on our proposal distribution (Section~\ref{sec:proposal}): $\eta$
balances the weight on frequent/infrequent neighbors and $\rho$ balances
the weight on sub-neighbors/super-neighbors within the frequent neighbors.
Note that the proposal
distribution does not affect the correctness of the MH sampling, but it does
affect the speed of convergence. Here the impact of $\eta$ is measured by
the average accept rate in the entire walk, i.e., the rate that a proposed
pattern is accepted on average.
Since frequent patterns have exponentially large probability mass to be
sampled, a larger value of $\eta$ should be desired. This is reflected in
Figure~\ref{fig:eta}, in which the average accept rate increases from about 35\%
when $\eta=0.4$ to more than 60\% at $\eta=0.9$. The other parameter $\rho$
controls the probability mass of sub-neighbors given that a frequent pattern
will be proposed. In graph pattern mining, the smaller graphs usually have
larger support. Therefore a $\rho$ of at least 0.5 is preferred, which can
be seen by the drop of average accept rate from 60\% to less than 40\% when
$\rho$ decreases from 0.5 to 0.4 in Figure~\ref{fig:gamma}. Interestingly,
as we deviate away from 0.5, the acceptance rate slowly drops and
adversely affects the sampling performance. This is because a balanced
sub-neighbor/super-neighbor proposal allows faster transition from one
pattern to another, making the chain well mixed instead of lingering in a
local region.

\para{Convergence analysis.} A decision we have to make is when to stop the
random walk and output a sample. 
In Section~\ref{sec:converge} we introduced
Z-score based Geweke diagnostic,
which compares the distribution at the beginning and end
of the chain. Since MCMC is typically used to
estimate a function of the underlying random variable instead of structural
data like graphs, we need to choose some properties of the patterns
which we will monitor using the Geweke test. The three metrics we use in the
experiment are the number of neighbors $N(x)$, the number of frequent neighbors
$N_1(x)$ and the number of nodes in the pattern $|x|$.
Figure~\ref{fig:converge} shows the convergence traces of a sample run
with $K=20$ and $\eps=0.5$ on the {\em DTP} dataset.
Each curve corresponds to the Z-score of a chain over the number of iterations.
It can be seen that the Markov chain we design has pretty fast convergence
rate thanks to the tuning of the proposal distribution.
For each chain, convergence is declared when the Z-scores of
all three metrics have fallen within the $[-1,1]$ range for 20 iterations continuously.
In Figure~\ref{fig:converge}, this happens around 150 iterations for most chains.

\section{Related Work}\label{sec:related}
In a broad sense, our paper belongs to the general problem of
privacy-preserving data mining - a topic that has been studied
extensively for a decade because of its numerous applications
to a wide variety of problems in the literature. A general overview
of various research works on this topic can be found in \cite{aggarwal2008privacy}.
Below we briefly review the results relevant to this paper.

\para{Data Mining with Differential Privacy.}
Ever since differential privacy \cite{Dwork2006a} was proposed and embraced by the database
community, the privacy requirement that various works try to achieve has
shifted from syntactic models like $k$-anonymity \cite{sweeney02-2} to the more rigorous
model of differential privacy. A formal introduction to differential privacy
can be found in Section~\ref{subsec:dp}.
There exist two basic approaches to differentially
private data mining. In the first approach, the data owner releases an
anonymized version of the original dataset under differential privacy.
And the user has the freedom of conducting any data mining task on the anonymized
dataset. We call this the `publishing model'. Examples include releasing
anonymized version of contingency tables \cite{Barak2007,Xiao2010},
data cubes \cite{Ding2011} and spatial data \cite{Cormode2012}.
The general idea in these work is to release tables of noisy
counts (histograms) and study how to ensure they are
sufficiently accurate for different query workloads.
In the other approach, differential privacy is applied to a specific data
mining task, such as decision tree induction \cite{Friedman2010},
social recommendations \cite{machana2011personalized} and
frequent itemset mining \cite{Thakurta2010}.
The problem addressed in this paper falls into this category.
In these works, randomness is often injected
to the intermediate results or sub-procedures of a mining algorithm.
While the output of the first approach
is more versatile, the second approach often leads to better utility (for specific data mining tasks)
since privacy-preserving techniques are
particularly designed for that data mining algorithm.

\para{Privacy-Protection of Graphs.}
The aforementioned works on differentially private data mining all deal with
structured data (tables or set-valued data).
For graph data, there are research efforts \cite{aggarwal2008privacy} to
anonymize a social network graph to prevent node and edge re-identification.
But most of them focus on modifying the graph structure to satisfy
$k$-anonymity, which has been proved to be insufficient \cite{aggarwal2008privacy}.
Recently, several works \cite{Karwa2011,Hay2009} emerge to provide differentially private analysis of
graph data, which releases some statistics such as the number of triangles
about a single (large) graph. Two types of differential privacy have been
introduced to handle graph data: node differential privacy and edge differential
privacy. It is still open whether any nontrivial graph statistics can be
released under node differential privacy due to its inherent large sensitivity
(e.g., removing
a node in a star graph may result in an empty graph). Hay {\em et al.} \cite{Hay2009} consider the
problem of releasing the degree distribution of a graph under a variant of
edge differential privacy. More recently, Karwa {\em et al.} \cite{Karwa2011} propose algorithms
to output approximate answers to subgraph counting queries, i.e., given a
query graph $H$, returning the number of
edge-induced isomorphic copies of $H$ in the input graph. The technique they
use is to calibrate noise according to the smooth sensitivity \cite{Nissim2007} of $H$
in the input graph. Karwa et al. The cases when $H$ is triangle,
$k$-star or $k$-triangle are studied in \cite{Karwa2011}.
Unfortunately, their work does not support the case when $H$ is an arbitrary subgraph yet.

In contrast, we have a different problem setting from \cite{Karwa2011} in this paper. First, like
\cite{Thakurta2010}, our privacy-preserving algorithm is associated with a specific and more complicated data
mining task. Second, we consider a graph database containing a collection of graphs
related to individuals. The only work we can find on privacy protection for a
graph database is \cite{Li2011}, which follows the `publishing model'.
Their goal is to achieve $k$-anonymity by
first constructing a set of super-structures and then generating synthetic
representations from them.

\para{Graph Pattern Mining.} Finally, we briefly discuss relevant works
on traditional non-private graph pattern mining.
A more comprehensive survey can be found in \cite{Data2010}.
Earlier works which aim at finding all the frequent patterns in a graph database
usually explore the search space in a certain manner. Representative approaches
include {\em a priori}-based
(e.g. \cite{inokuchi2000apriori}) and pattern growth based
(e.g. gSpan \cite{yan2002gspan}). An issue with this
direction is that the search space grows exponentially with the pattern size,
which may reach a computation bottleneck.
Thus later works aim at mining
{\em significant} or {\em representative} patterns with scalability. One way of achieving
this is through random walk \cite{AlHasan2009}, which also motivates our use of MCMC sampling
for privacy preserving purpose. Another remotely related work is \cite{williams2010probabilistic},
which connects probabilistic inference and differential
privacy. It differs from this work by focusing on inferencing on the output of a
differentially private algorithm.

\section{Concluding Remarks}\label{sec:conclusion}
We have presented a novel technique for differentially private
mining of frequent graph patterns.
The proposed solution integrates the process
of graph mining and privacy protection into an MCMC sampling framework.
We have explored the design space of the proposal distribution
and the score function and their impact on the performance of the algorithm.
Moreover, we have established the theoretical privacy and utility guarantee
of our algorithm.
An efficient algorithm for counting the neighbors of a pattern has been
proposed to greatly reduce the time-consuming subgraph isomorphism tests.

Experiments on both synthetic and real datasets show that with moderate
amount of privacy budget, {\em Diff-FPM} is able to output frequent patterns
with over 80\% precision and support accuracy. We also notice the drop in
utility with the increase of the number of outputs or the decrease in dataset
size, which is inevitable under the requirement of differential privacy.

{
\bibliographystyle{abbrv}

}
\vspace{1.5em}
\end{document}